%% file: psp.tex
\documentclass[11pt]{article}
\usepackage[compact]{titlesec}

\usepackage{times}
\usepackage{latexsym}
\usepackage{amsfonts,amsthm,amssymb}
\usepackage{amsmath}
\usepackage{euscript}
\usepackage{amstext}
\usepackage{graphicx}
\usepackage{color}
\usepackage{url,hyperref}
\usepackage{verbatim}
\usepackage{xspace}
\usepackage{algorithm2e}
\usepackage{framed}

\newtheorem{lemma}{Lemma}[section]
\newtheorem{theorem}[lemma]{Theorem}

\newtheorem{definition}[lemma]{Definition}
\newtheorem{corollary}[lemma]{Corollary}
\newtheorem{proposition}[lemma]{Proposition}

\newenvironment{proofof}[1]{\smallskip\noindent{\bf Proof of #1}}%
        {\hspace*{\fill}$\Box$\par}

\newcommand{\eps}{\epsilon}

\newcommand{\rr}{\textrm{\sc  RR}\xspace}

\newcommand{\laps}{\textrm{\sc  LAPS}\xspace}

\newcommand{\pf}{\textrm{\sc PF}\xspace}

\newcommand{\etal}{et al.\ }

\newcommand{\dd}{\texttt{d}}

\newcommand{\cR}{{\cal R}}

\newcommand{\cS}{{\cal S}}
\newcommand{\cF}{{\cal F}}
\newcommand{\cI}{{\cal I}}

\newcommand{\oneindi}{{\texttt 1}}

\newcommand{\initOneLiners}{%
    \setlength{\itemsep}{0pt}
    \setlength{\parsep }{0pt}
    \setlength{\topsep }{0pt}
}

\newcommand{\cL}{\mathcal L}
\newcommand{\cA}{\mathcal A}

\newcommand{\cC}{\mathcal C}
\newcommand{\cT}{\mathcal T}

\newcommand{\bone}{{\texttt 1}}

\newcommand{\primal}{\mathsf{PRIMAL}}
\newcommand{\dual}{\mathsf{DUAL}}

\newcommand{\cppf}{\mathsf{CP_{PF}}}

\newcommand{\vx}{\vec{x}}

\newcommand{\betav}{\vec{\beta}}

\newcommand{\xv}{\vec{x}}

\newcommand{\yv}{\vec{y}}
\newcommand{\zerov}{\vec{0}}

\newcommand{\Bbf}{{\bf B}}
\newcommand{\poly}{\mathcal{P}}

\newcommand{\psp}{\textsf{PSP}\xspace}
\newcommand{\qv}{\vec{f}}

\renewcommand{\vec}[1]{\mathbf{#1}}

\newcommand{\lpp}{\mbox{$\mathsf{LP}_\mathsf{primal}$}\xspace}
\newcommand{\lpd}{\mbox{$\mathsf{LP}_\mathsf{dual}$}\xspace}

\newcommand{\slaps}{\textrm{\sc  S-LAPS}\xspace} 

\begin{document}

\title{Competitive Algorithms from Competitive Equilibria: Non-Clairvoyant Scheduling under Polyhedral Constraints}
\author{Sungjin Im\thanks{Electrical Engineering and Computer Science, University of California, Merced CA 95344. {\tt sim3@ucmerced.edu}. This work was done while the author was at Duke. Supported by NSF Award CCF-1008065.}  \and Janardhan Kulkarni \thanks{Department of Computer Science, Duke University , 308 Research Drive, Durham, NC 27708. {\tt kulkarni@cs.duke.edu}. Supported by NSF Awards CCF-1008065 and IIS-0964560.}  \and
Kamesh Munagala\thanks{Department of Computer Science, Duke University, Durham NC 27708-0129. {\tt kamesh@cs.duke.edu}. Supported by an award from Cisco, and by NSF  grants CCF-0745761, CCF-1008065, CCF-1348696, and IIS-0964560.}
}
\date{}
\maketitle

\thispagestyle{empty}

\begin{abstract}
We introduce and study a general scheduling problem that we term the Packing Scheduling problem (\psp). In this problem, jobs can have different arrival times and sizes; a scheduler can process job $j$ at rate $x_j$, subject to arbitrary packing constraints over the set of rates ($\vec{x}$) of the outstanding jobs. The \psp framework captures a variety of scheduling problems, including the classical problems of unrelated machines scheduling, broadcast scheduling, and scheduling jobs of different parallelizability. It also captures scheduling constraints arising in diverse modern environments ranging from individual computer architectures to data centers. More concretely, \psp models multidimensional resource requirements and parallelizability, as well as network bandwidth requirements found in data center scheduling.

\smallskip	
In this paper, we design {\em non-clairvoyant} online algorithms for \psp and its special cases -- in this setting, the scheduler is unaware of the sizes of jobs. Our results are summarized as follows. 
	
\begin{itemize}
\item For minimizing total weighted completion time, we show a $O(1)$-competitive algorithm.  Surprisingly, we achieve this result by applying the well-known Proportional Fairness algorithm (\pf) to perform allocations each time instant. Though \pf has been extensively studied in the context of maximizing fairness in resource allocation, we present the {\em first}  analysis in adversarial and general settings for optimizing job latency. Our result is also the first $O(1)$-competitive algorithm for weighted completion time for several classical non-clairvoyant scheduling problems. 
\item For minimizing total weighted flow time, 
for any constant $\eps >0$, any $O(n^{1-\eps})$-competitive algorithm requires extra speed (resource augmentation) compared to the offline optimum. 
We  show that \pf is a $O(\log n)$-speed $O(\log n)$-competitive non-clairvoyant algorithm, where $n$ is the total number of jobs. We further show that there is an instance of \psp for which no non-clairvoyant algorithm can be $O(n^{1-\eps})$-competitive with $o(\sqrt{ \log n} )$ speed. 

\item For the classical problem of minimizing total flow time for unrelated machines in the non-clairvoyant setting, we present the first online algorithm which is scalable ($(1+\eps)$-speed $O(1)$-competitive for any constant $\eps > 0$). No non-trivial results were known for this setting, and the previous scalable algorithm could handle only related machines. We develop new algorithmic techniques to handle the unrelated machines setting that build on a new single machine scheduling policy. Since unrelated machine scheduling is a special case of \psp, when contrasted with the lower bound for \psp, our result also shows that \psp is significantly harder than perhaps the most general classical scheduling settings. 
\end{itemize}

Our results  for \psp show that instantaneous fair scheduling algorithms can also be effective tools for minimizing the overall job latency, even when the scheduling decisions are non-clairvoyant and constrained by general packing constraints. 
\end{abstract}

\newpage
\setcounter{page}{1}

\input{1.intro.tex}

\input{2.algorithm.tex}
\input{3.completion.tex}

\input{5.unrelated.tex}

\input{6.conclusion.tex}

\bibliographystyle{plain}
\bibliography{mmf}
\appendix
\input{4.flowtime-ub.tex}

\input{4.flowtime-lb.tex}



\end{document}

%% file: 1.intro.tex
\section{Introduction}
	\label{sec:intro}

Consider a typical data center setting,  where there is a cluster of machines with a distributed file system implementation (such as HDFS~\cite{hdfs}) layered on top of the cluster. Users submit executables (or jobs) to this cluster. In a typical {\sc MapReduce} implementation such as Hadoop~\cite{hadoop}, each job is a collection of parallel map and reduce tasks requiring certain CPU, disk space, and memory to execute. The job therefore comes with a request for resources in each dimension; these can either be explicitly specified, or can be estimated by the task scheduler from a high-level description of the job.

In a general scheduling scenario that has gained a lot of attention recently (see~\cite{drf} and followup work~\cite{ColeGG,Zaharia08,Ahmad2012,ec2spot,popa2012faircloud,lee2011heterogeneity}), there are $M$ different types of resources. In the context of a data center, these could be CPU, disk, memory, network bandwidth, and so on. The resources are assumed to be infinitely divisible due to the abundance of resources, and there is $R_d$ amount of resource $d$. 

Each job $j$ is associated with resource demand vector $\qv_j = (f_{j1}, f_{j2}, ..., f_{jM})$ so that it requires $f_{jd}$ amount of the $d^{th}$ resource.  At each time instant, the resources must be feasibly allocated among the jobs. If job $j$ is allocated resource vector $(a_{j1}, a_{j2}, \ldots, a_{jM})$ where $a_{jd} \le f_{jd}$,  it is processed at a rate that is determined by its bottleneck resource, so that its rate is $x_j = \min_d (a_{jd}/f_{jd})$. Put differently, the rate vector $\vec{x}$ needs to satisfy the set of packing constraints:
$$ \poly=  \left \{\sum_j x_j f_{jd} \le R_d \ \ \forall d \in [M];  \qquad \vec{x} \le \vec{1}; \qquad \vec{x} \ge 0 \right\}$$

The above resource allocation problem, that we term \textsf{Multi-dimensional Scheduling} is not specific to data centers -- the same formulation has been widely studied in network optimization, where resources correspond to bandwidth on edges and jobs correspond to flows. The bandwidth on any edge must be feasibly allocated to the flows, and the rate of a flow is determined by its bottleneck allocation. For instance, see~\cite{KMT} and copious followup work in the networking community. 

The focus of such resource allocation has typically been instantaneous {\em throughput}~\cite{drf}, {\em fairness}~\cite{drf,popa2012faircloud,lee2011heterogeneity}, and {\em truthfulness}~\cite{drf,ColeGG} -- at each time instant, the total rate must be as large as possible, the vector $\vec{x}$ of rates must be ``fair" to the jobs, and the jobs should not have incentive to misreport their requirements. The scheduling (or temporal) aspect of the problem has largely been ignored. Only recently, in the context of data center scheduling, has {\em response time} been considered as an important metric -- this corresponds to the total completion time or total flow time of the jobs in scheduling parlance. Note that the schedulers in a data center context typically have access to instantaneous resource requirements (the vectors $\vec{f_j}$), but are not typically able to estimate how large the jobs are in advance -- in scheduling parlance, they are {\em non-clairvoyant}. They further are only aware of jobs when they arrive, so that they are {\em online} schedulers.

Though there has been extensive empirical work measuring response times of various natural resource allocation policies for data center scheduling~\cite{drf,Zaharia08,Ahmad2012,popa2012faircloud,lee2011heterogeneity}, there has been very little theoretical analysis of this aspect; see~\cite{Bonald06, KellyMW09} for recent queueing-theoretic analysis of network routing policies. This is the starting point of our paper -- we formalize non-clairvoyant, online scheduling under packing constraints on rates as a general framework that we term \textsf{General Polytope Scheduling Problem} (\psp), and present competitive algorithms for problems in this framework.

\subsection{General Polytope Scheduling Framework}
In this paper, we consider a generalization of the multi-dimensional scheduling problem discussed above. In this framework  that we term \textsf{General Polytope Scheduling Problem} (\psp), the packing constraints on rates can be arbitrary. We show below (Section~\ref{sec:app}) that this framework not only captures multi-dimensional scheduling, but also captures classical scheduling problems such as unrelated machine scheduling (with preemption and migration), fractional broadcast scheduling, as well as scheduling jobs with varying parallelizability -- only some special cases have been studied before. 

\medskip

In \psp, a scheduling instance consists of $n$ jobs, and each job $j$ has weight $w_j$, size $p_j$, and arrives at time $r_j$. At any time instant $t$, the scheduler must assign rates $\{x_j\}$ to the current jobs in the system. Let $x^\cA_j(t)$ denote the rate at which job $j$ is processed at time $t$ by a scheduler/algorithm $\cA$. Job $j$'s completion time $C^\cA_j$ under the schedule of $\cA$ is defined to be the first time $t'$ such that $\int_{t = r_j}^{t'} x^\cA_j(t) \dd t \geq p_j$. Similarly, we define job $j'$ flow time as $F^\cA_j = C^\cA_j - r_j$, which is the length of time job $j$ waits to be completed since its arrival.  When the algorithm $\cA$ and time $t$ are clear from the context, we may drop them from the notation. 

We assume the vector of rates $\vec{x}$ is constrained by a packing polytope $\poly$, where the matrices $H, Q$ have non-negative entries.
\begin{equation} \label{eq:main}
 \poly  =  \Big\{ \vec{x} \ \ | \ \ \vec{x} \le Q \vec{z};  \qquad  H \vec{z} \le \vec{1};  \qquad \vec{x} \ge 0;  \qquad \vec{z} \ge 0 \Big\}
\end{equation}

The class of scheduling algorithms we consider are constrained by several properties, all of which are naturally motivated by modern scheduling applications. 
\begin{itemize}
\item It is  {\em online} and learns about job $j$ only when it arrives. Before this point, $x_j = 0$.
\item It is {\em non-clairvoyant}, {\em i.e.}, does not know a job's size $p_j$ until completing the job.  
\item It is allowed to re-compute $\vec{x}(t)$ at any real time $t$ arbitrarily often. As we will see below, this allows for pre-emption as well as migration across machines at no cost. Though we technically allow infinitely many re-computations, our algorithms will  perform this computation only when jobs either arrive or complete.
\end{itemize}

Without loss of generality, we will assume the matrices $H,Q$ are known in advance to the scheduler and are independent of time, so that $\poly$ itself is time-invariant. One way of enforcing this is to assume that jobs arrive online from a subset of a (possibly countably infinite) universe $U$ of possible jobs, and the matrices $H,Q$ are defined over this universe. This is purely done to simplify our description and notation -- in our applications, the polytope $\poly$ will indeed be defined only over the subset of jobs currently in the system, and the algorithms we design will make no assumptions over future jobs.   

Under these assumptions, we will investigate non-clairvoyant online algorithms that minimize the overall job latency, {\em i.e.},  the total weighted completion time $\sum_j w_j C_j$ (resp. total weighted flow time $\sum_j w_j F_j$). We will compare our algorithm against the optimal offline scheduler that knows the scheduling instance ($w_j, p_j, r_j$ for all jobs $j$) in advance, using the standard notion of competitive ratio. 

\medskip
Our main result is the {\em first} constant competitive non-clairvoyant algorithm for \psp under the weighted completion time metric, which also implies the first such result for all the applications we consider below. Our algorithm is in fact a natural and widely studied method of fair allocation termed {\em proportional fairness} (\pf). In effect, we show that a resource allocation approach to viewing scheduling problems yields insights into designing competitive non-clairvoyant schedules. (We show several other results; see Section~\ref{sec:results}.)

\subsection{Applications of the \psp Framework} 
\label{sec:app}
Before discussing our results in depth, we present several concrete problems that fall in the \psp framework. In each case, we present a mapping to the constraints in $\poly$. We have already seen the special case of multi-dimensional scheduling.  We note that our framework can handle combinations of these problems as well.

\medskip
\noindent {\bf All-or-nothing Multidimensional Scheduling.} 

In multidimensional scheduling, we have assumed that a job needs all resources to execute, and given a fraction of all these resources, it executes at a fraction of the rate. However, in practice, a job often needs to receive its entire requirement in order to be processed~\cite{vmpack, Zaharia08} -- this can be necessitated by the presence of indivisible virtual machines that need to be allocated completely to jobs. Therefore,  a job $j$ is processed at a rate of 1 when it receives the requirement $\vec{f}_j$, otherwise not processed at all.  This all-or-nothing setting was studied recently in~\cite{FoxK13} when there is only one dimension. 
To see how this problem is still captured by \psp, define variables that encode feasible schedules. Let $\cS$ denote the collection of subsets of jobs that can be scheduled simultaneously. Let $z_S$ denote the indicator variable which becomes 1 if and only if $S$ is exactly the set of jobs currently processed. We observe this setting is captured by the following polytope.

\begin{equation}
 \poly = \Big \{x_j \leq \sum_{S: j \in S} z_S  \ \ \forall j; \qquad \sum_{S \in \cS}  z_S \leq 1; \qquad \vec{x} \ge 0 ; \qquad \vec{z} \ge 0 \Big \}
\end{equation}

The solution to $\poly$ is a set of preemptive schedules that process jobs in $S$ for $z_S$ fraction of time.

\medskip
\noindent {\bf Scheduling Jobs with Different Parallelizability over Multiple Resources.} 

In most cluster computing applications, a job is split into several {\em tasks} that are run in parallel. However, jobs may have different parallelizability depending on how efficiently it can be decomposed into tasks \cite{wolf2010flex}. To capture varying degree of parallelizability, an elegant theoretical model a.k.a. arbitrary speed-up curves was introduced by Edmonds \etal\cite{EdmondsCBD03}. In this model, there is only one type of resources, namely homogeneous machines, and a job $j$ is processed at a rate of $\Gamma_{j}(m_j)$ when assigned $m_j$ machines. The parallelizability function $\Gamma_j$ can be different for individual jobs $j$, and is assumed to be non-decreasing, and sub-linear ($\Gamma_{j}(m_j)/m_j$ is non-increasing). Due to the simplicity and generality, this model has received considerable amount of attention \cite{RobertS08,ChanEP09, EdmondsIM11,  EdmondsP12, FoxIM13}.  However, no previous work addresses parallelizability in multiple dimensions and heterogeneous machines. Here we extend $\Gamma_j$ to be a multivariate function that takes the resource vector ${\vec z}_j:= (z_{j1}, z_{j2}, ..., z_{jM})$ of dimension $M$ job $j$ is assigned, and outputs the maximum speed job $j$ can get out of the assignment. The function $\Gamma_j$ is restricted to be concave in any positive direction. Observe that $x_j \leq \Gamma_j( {\bf z}_j)$ can be (approximately) expressed by a set of packing constraints over ${\bf a}_j$ that upper bound $x_j$. (The \psp framework can be generalized to a convex polytope, and our results carry over). Then the obvious extra constraints is  $\sum_j  {\bf z}_j \leq  {\bf 1}$.  This extension can also capture tradeoff between resources (complements or substitutes) that can be combinatorial in nature. For example, a job can boost its execution by using more CPU or memory in response to the available resources.

\medskip
\noindent {\bf Non-clairvoyant Scheduling for Unrelated Machines.} In this problem there are $M$ unrelated machines. Job $j$ is processed at rate $s_{ij}  \in [0, \infty)$ on each machine $i$. (Unrelated machines generalize related machines where machines have different speeds independent of jobs). The online algorithm is allowed to preempt and migrate jobs at any time with no penalty -- without migration, any online algorithm has an arbitrarily large competitive ratio for the total completion time \cite{GuptaIKMP12}. The important constraint is that at any instantaneous time, each machine can schedule only one job, and a job can be processed only on a single machine. 

We can express this problem as a special case of $\psp$ as follows. Let $z_{ij}$ denote the fraction of job $j$ that is scheduled on machine $i$. Then:
$$ \poly = \Big \{ x_j \le \sum_i s_{ij} z_{ij}\ \ \forall j; \qquad \sum_j z_{ij} \le 1 \ \ \forall i; \qquad \sum_i z_{ij} \le 1  \ \ \forall j; \qquad \vec{x} \ge 0 ; \qquad \vec{z} \ge 0 \Big \}$$
Note that any feasible $\vec{z}$ can be decomposed into a convex combination of injective mappings from jobs to machines preserving the rates of all jobs. 
Therefore, any solution to $\poly$ can be feasibly scheduled with preemption and reassignment. As before, the rates $\vec{s_j}$ are only revealed when job $j$ arrives. 
No non-trivial result was known for this problem before our work. The only work related to this problem considered the setting where machines are related and jobs are unweighted~\cite{GuptaIKMP12}. The algorithm used in~\cite{GuptaIKMP12} is a variant of Round Robin; however, as pointed out there, it is not clear how to extend these techniques to take job weights and heterogeneity of machines into account, and this needs fundamentally new ideas. 

\medskip
\noindent {\bf Generalized Broadcast Scheduling.} There are $M$ pages of information (resources) that is stored at the server. The server broadcasts a unit of pages at each time step. When a page $i$ is broadcast, each job $j$ (of total size $p_j$) is processed at rate  $s_{ij}$. The vector $\vec{s_j}$ of rates  is only revealed when job $j$ arrives.  Therefore:
$$ \poly = \Big\{ x_j \le \sum_{i \in [M]} s_{ij} z_i \ \ \forall j; \qquad \sum_{i \in [M]} z_i \le 1; \qquad \vec{x} \ge 0; \qquad \vec{z} \ge 0 \Big \}$$
This setting strictly generalizes classical fractional  broadcast scheduling where it is assumed that for each job $j$, the rate $s_{ij} = 0$ for all pages except one page $i$, and for the page $i$, $s_{ij} = 1$. In general, $s_{ij}$ can be thought of as measuring how much service $i$ makes happy client $j$ -- for motivations, see~\cite{AzarG11, ImNZ12} where more general submodular functions were considered for clairvoyant schedulers in a different setting. 
We note that fractional classical broadcast scheduling is essentially equivalent to the integral case since there is an online rounding procedure~\cite{BansalKN10} that makes the fractional solution integral while increasing each job's flow time by at most a constant factor (omitting technicalities).  The unique feature of broadcast scheduling is that there is no limit on the number of jobs that can be processed simultaneously as long as they ask for the same resource. It has therefore received considerable attention in theory \cite{GandhiKPS06, BansalCS08, BansalKN10, EdmondsIM11, ImM12, BansalCKL14} and has abundant applications in practice such as multicast systems, LAN and wireless systems \cite{Wong88,AcharyaFZ95,AksoyF98}.

\subsection{Our Algorithms and Results}
\label{sec:results}
Our main result is the following; it also yields the first such result for {\em all} the applications discussed in Section~\ref{sec:app} above.
\begin{theorem}
	\label{thm:completion}[Section~\ref{sec:completion}.]
For the weighted completion time objective, there exists a $O(1)$-competitive non-clairvoyant scheduling algorithm for $\psp$.
\end{theorem}

We show this result by a simple algorithm that has been widely studied in the context of fairness in resource allocation, dating back to Nash~\cite{bargaining}. This is the Proportional Fairness (\pf) algorithm~\cite{bargaining,KMT,drf}. Let $\cA_t$ denote the set of jobs alive at time $t$.  At time $t$, the rates are set using the solution to the following convex program (See Section~\ref{sec:algorithm-pf} for more details).
$$ \vec{x^*}(t) = \mbox{argmax} \Big\{ \sum_{j \in \cA_t} w_j \log x_j \ \ | \ \  \vec{x} \in \poly \Big\}$$

To develop intuition, in the case of multi-dimensional scheduling with resource vector $\vec{f}_j$ for job $j$, the \pf algorithm implements a {\em competitive equilibrium} on the jobs. Resource $d$ has price $\lambda_d$ per unit quantity. Job $j$ has budget $w_j$, and sets its rate $x_j$ so that it spends its budget, meaning that $x_j = \frac{w_j}{\sum_d \lambda_d f_{jd}}$. The convex program optimum guarantees that there exists a set of prices $\{\lambda_d\}$ so that the market clears, meaning that all resources with non-zero price are completely allocated. 

In the same setting, when there is $K=1$ dimension, the \pf solution reduces to {\em Max-Min Fairness} -- the resource is allocated to all jobs at the same rate (so that the increase in $f_j x_j$ is the same), with jobs dropping out if $x_j = 1$. Such a solution makes the smallest allocation to any job as large as possible, and is fair in that sense.  Viewed this way, our result seems intuitive -- a competitive non-clairvoyant algorithm needs to behave similarly to round-robin (since it needs to hedge against unknown job sizes), and the max-min fair algorithm implements this idea in a continuous sense. Therefore, fairness seems to be a requirement for competitiveness. However this intuition can be misleading -- in a multi-dimensional setting, not all generalizations of max-min fairness are competitive -- in particular, the popular Dominant Resource Fair (DRF) allocation and its variants~\cite{drf} are $\omega(1)$ competitive. Therefore, though fairness is a requirement, not all fair algorithms are competitive. 

Multidimensional scheduling is not the only application where the ``right"  notion of fairness is not clear. As discussed before, it is not obvious how to generalize the most intuitively fair algorithm Round Robin (or Max-Min Fairness) to unrelated machine scheduling -- in~\cite{GuptaIKMP12}, a couple of natural extensions of Round Robin are considered, and are shown to be $\omega(1)$-competitive for total weighted completion time. In hindsight, fairness was also a key for development of online algorithms in broadcast scheduling~\cite{BansalKN10}. Hence, we find the very existence of a unified, competitive, and fair algorithm for  \psp quite surprising!

\paragraph{Flow Time Objective.} We next consider the weighted flow time objective for \psp. We note that even for classical single machine scheduling, any deterministic algorithm is $\omega(1)$-competitive \cite{BansalC09}. Further, in the unrelated machine setting, there is no online algorithm with a bounded competitive ratio~\cite{GargK07}. 
Hence to obtain positive results, we appeal to speed augmentation which is a popular relaxation of the worst case analysis framework for online scheduling~\cite{kirk}. Here, the online algorithm is given speed $s \ge 1$, and is compared to an optimal scheduler which is given a unit speed. 
More precisely, we compare our algorithm against an optimal omniscient solution which is constrained by the tighter constraint $H \vec{z} \le \frac{1}{s}$.

\begin{theorem}
	\label{thm:flow} [Appendix~\ref{sec:flowtime}.]
For \psp, the PF algorithm is $O(\log n)$-speed, $O(\log n)$-competitive for minimizing the total weighted flow time.  Furthermore, there exists an instance of \psp for which no deterministic non-clairvoyant algorithm is $O(n^{1- \eps})$-competitive for any constant $0 < \eps < 1$ with $o(\sqrt{ \log n})$-speed.
\end{theorem}

We note that this is the first non-trivial flow time result for all the applications mentioned in Section~\ref{sec:app}.  

\paragraph{Unrelated Machine Scheduling.} We finally consider (the special case of) non-clairvoyant scheduling on heterogeneous machines. Recall that here, each job $j$ is processed at a rate of $s_{ij}$ on machine $i$. 
In this case, the above results show $O(1)$-competitive algorithms for total completion time, and $O(\log n)$-speed, $O(\log n)$-competitive algorithm for total flow time. We improve these results to show the {\em first} scalable ($(1+\eps)$-speed $O(1)$-competitive) algorithm for the total (unweighted) flow time objective. 

\begin{theorem} 
	\label{thm:flow-no-scalable} [Section~\ref{sec:unrelated}.]
	For any $\eps >0$, there is a $(1+\eps)$-speed $O(1/ \eps^2)$-competitive non-clairvoyant for the problem of minimizing the total (unweighted) flow time on unrelated machines.
\end{theorem}

We show this result by developing a new algorithm that we term BLASS ({\em Balanced Latest Arrival Smooth Scheduling}) that assigns jobs to machines based on delays caused to other jobs. On each machine, jobs are scheduled smoothly favoring recent jobs over older jobs. When jobs depart, it rearranges jobs using arrival order as a priority rule, in order to balance the objective across machines. 

\subsection{Our Techniques}
All our analysis is based on dual fitting. Dual fitting is popular for design and analysis of approximation and online algorithms, but two elegant works in \cite{AnandGK12, GuptaKP12} initiated dual fitting approach for online scheduling. Since our focus is on a linear objective, the total (weighted) flow time, and \cite{GuptaKP12} is concerned with non-linear objectives, we mainly discuss the work in \cite{AnandGK12}, and compare it to our work. The work in \cite{AnandGK12} considered (clairvoyant) unrelated machine scheduling for the total weighted flow time. Their
approach formulates the natural LP relaxation for weighted flow time, and sets feasible dual variables of this program so that the dual objective is within a constant of the primal objective. Their algorithm couples natural single-machine scheduling policies with a greedy rule that assigns each arriving job to a machine that increases the objective the least assuming that no more jobs arrive. The algorithm is immediate-dispatch and non-migratory -- it immediately assigns an arriving job to a machine, and the job never migrates to other machines. However, such nice properties  require that the algorithm should be clairvoyant. In fact, there is a simple example that shows that any non-clairvoyant and immediate dispatch algorithm has an arbitrarily large competitive ratio if migration is not allowed~\cite{GuptaIKMP12}. Ironically, migration, which seems to give more flexibility to the algorithm, makes the analysis significantly more challenging. Hence it is no surprise that essentially all online algorithms for heterogeneous machine scheduling have been non-migratory -- the only exception being~\cite{GuptaIKMP12}, which gives a scalable algorithm for related machines for the flow time objective. For the same reasons, there has been very little progress in non-clairvoyant heterogeneous machine scheduling which is in sharp contrast to the recent significant progress in the clairvoyant counterpart \cite{ChadhaGKM09, AnandGK12}.

Since \psp captures non-clairvoyant scheduling on unrelated machines, the algorithms need to be migratory. Since migration disallows reduction to single machine scheduling, this precludes the types of dual variable settings considered in~\cite{AnandGK12}. To develop intuition, in dual fitting, we are required to distribute the total weight of unsatisfied (resp. alive) jobs to the dual variables corresponding to constraints in $\poly$. We therefore connect the dual values found by the KKT condition to the dual variables of the completion (resp. flow) time LP for \psp. This is a challenging task since the duals set by KKT are obtained by instantaneous (resource allocation) view of \pf while the duals in the LP should be globally set considering each job's completion time. For the completion time objective we manage to obtain $O(1)$-competitiveness by reconciling these two views using the fact that the contribution of the unsatisfied jobs to the objective only decreases over time. 
For flow time, such a nice structure is elusive, and in fact, we show that any online deterministic non-clairvoyant algorithm has a large competitive ratio when given speed less than $o(\sqrt{\log n})$. The lower bound is constructed by carefully hiding ``big" jobs in multiple layers.  However, we show that $\pf$ is $O(\log n)$-competitive when given $O(\log n)$-speed. At a high-level this is achieved by 
decomposing $\pf$'s schedule to a sequence of completion time analyses, and combining these dual variables at the end. 

For the scalable algorithm for unrelated machine scheduling (Theorem~\ref{thm:flow-no-scalable}), we note that the previous work~\cite{GuptaIKMP12} could handle only the case where machines have different speeds independent of jobs. Our algorithm BLASS is entirely different from that in~\cite{GuptaIKMP12}. We introduce a new technique for {\em rearranging} jobs when a job departs. This rearrangement procedure considers jobs in increasing order of arrival time, and performs local optimization of a job's processing assuming a hypothetical Round Robin scheduling policy only on jobs with earlier arrival times.  Our main technical contribution is to show that the rearrangement procedure maintains a certain global optimality property about which machine a job is scheduled on. The actual scheduling policy on each machine is a new single machine scheduling policy called \slaps and is different from the proxy we perform the local optimization on. \slaps is an extension of Round Robin, which smoothly favors recent jobs over older jobs. We need this smoothness to achieve scalability in Round Robin type algorithms. The interplay of these two ideas is critical in showing the competitive ratio. We note that our smooth scheduling idea builds on the Latest Arrival Processor Sharing (\laps) algorithm~\cite{EdmondsP12}; however, we do not know how to apply their algorithm to unrelated machine scheduling since we cannot bound the delays jobs introduce to other jobs. We hope that our smooth variant, which makes the inter-job delays more transparent, finds more applications in  multiple machine scheduling.

\subsection{Related Work} 
	\label{sec:related}
We only summarize related work that have not been discussed before. We note that \psp is NP-hard even when all jobs arriving are known a priori -- this follows from  the well-known NP-hardness of the problem of minimizing the total weighted completion time on a single machine. In the offline setting, it is easy to obtain a $O(1)$-approximation for $\psp$ in the metric $\sum_j w_j C_j$. It can be achieved by LP rounding, for example, see~\cite{ImMP11}; similar ideas can be found in other literature \cite{SchulzS97, QueyranneS02}. Tight upper bounds have been developed for individual scheduling problems in completion time metric; see~\cite{WS} for a nice overview. In the online setting,~\cite{ChadhaGKM09, AnandGK12} give a scalable (clairvoyant) algorithm for the weighted flow time objective on unrelated machines. Linear (or convex) programs and dual fitting approaches have been popular for online scheduling; 
for an overview of online scheduling see~\cite{PruhsST}. Though~\cite{AzarBFP13} study a general online packing and covering framework, it does not capture temporal aspects of scheduling and is very different from our framework. Our work is also different from \cite{AzarBFP13} from the technical point of view. Our algorithm uses a natural algorithm \pf and dual fitting using KKT conditions while \cite{AzarBFP13} uses the multiplicative weights update method.

%% file: 2.algorithm.tex
\section{The Proportional Fairness (\pf) Algorithm and Dual Prices}
	\label{sec:algorithm-pf}
We first set up useful notation that will be used throughout this paper. We will refer to our algorithm \textsf{Proportional Fairness} (\pf) simply as $\cA$. We let $\cA_t:= \{j \; | \; r_j \leq t < C^{\cA}_j\}$ denote the set of outstanding/alive jobs at time $t$ in the algorithm's schedule. Similarly, let $U_t:= \{j \; | \; t < C^{\cA}_j\}$ denote the set of unsatisfied jobs. Note that $\cA_t \subseteq U_t$, and $U_t$ can only decrease as time $t$ elapses. We let $U_0$ denote the entire set of jobs that actually arrive.  We denote the inner product of two vectors $\vec{u}$ and $\vec{v}$ by $\vec{u} \cdot \vec{v}$. For a matrix $B$, $B_{i\cdot}$ denotes the $i^{th}$ row (vector) of matrix $B$. Likewise, $B_{\cdot i}$ denotes the $i^{th}$ column vector of matrix $B$. The indicator variable $\bone()$ becomes 1 iff the condition in the parentheses is satisfied, otherwise 0. 

As mentioned before, $C^\cA_j$ denotes job $j$'s completion time in $\cA$'s schedule. Let $F^\cA_j:= C^\cA_j - r_j$ denote job $j$'s flow time; recall that $r_j$ denotes job $j$'s release time.  For notational simplicity, we assume that times are {\em slotted}, and each time slot is sufficiently small compared to job sizes. By scaling, we can assume that each time slot has size 1, and we assume  that jobs arrive and complete only at integer times. These simplifying assumptions are w.l.o.g. and will make notation simpler.

To present our algorithm and analysis more transparently, we take a simpler yet equivalent view of the \psp by projecting the polytope $\poly$ into $\vec{x}$:
\begin{equation} \label{eq:simpler-main}
 \poly  = \left \{ B \vec{x} \le \vec{1};  \qquad \vec{x} \ge 0 \right\}, 
\end{equation}
\noindent where $B$ has no negative entries. The equivalence of these two expressions can be easily seen by observing that the definition in (\ref{eq:main}) is equivalent to that of general packing polytopes. We assume that $B$ has $D$ rows. 

Recall that $\cA_t:= \{j \; | \; r_j \leq t < C^{\cA}_j\}$ denotes the set of outstanding/alive jobs at time $t$ in our algorithm's schedule. At each time $t$ (more precisely, either when a new job arrives or a job is completed), the algorithm \textsf{Proportional Fairness} (\pf) solves the following convex program. 
\begin{align} 
	\max \sum_{j \in \cA_t} w_j \log x_j				  \tag{$\cppf$} \label{cppf}
\end{align}
\vspace{-5mm}
\begin{align*}   
&&	s.t. \quad	B \vec{x} &\le \vec{1}  && \\
&&			x_j  &= 0   && \forall j \notin \cA_t
\end{align*}
Then (\pf) processes each job $j$ at a rate of $x^*_{jt}$ where $x^*_{jt}$ is the optimal solution of the convex program at the current time $t$. Here the time $t$ is added to subscript since the scheduling decision changes over time as the set of outstanding jobs, $\cA_t$ does. For compact notation, we use a vector changing over time by adding $t$ to subscript -- for example,  $\vec{x}_t^*$ denotes the vector $\{x^*_{jt}\}_j$.  Observe that the constraint $\vx \geq \zerov$ is redundant since $x^*_j > 0$ for all $j \in \cA_t$. 

The dual of $\cppf$ has  variables  $y_d, d \in [D]$ corresponding to the primal constraints $B_{d\cdot} \cdot \vec{x} \leq 1$. Let $\vec{y}_t := ( y_{1t}, y_{2t}, ..., y_{Dt})$.   By the KKT conditions \cite{Boyd}, any optimal solution $\xv^*$ for $\cppf$ must satisfy the following conditions for some $\yv^*$:
\begin{align} 
&&	 	 y^*_{dt} \cdot (B_{d\cdot} \cdot \xv^*_t -1)  &= 0 &&\quad \forall t, d \in [D] \label{eqn:kkt-1}\\
&&		  \frac{w_j}{x^*_{jt}}  &=  B_{\cdot j} \cdot \yv^*_t  &&\quad \forall t, j \in \cA_t  \label{eqn:kkt-2} \\
&&		  \yv^*_t  &\geq  0  &&\quad \forall t  \label{eqn:kkt-3}
\end{align}

We emphasize that the new  definition of (\ref{eq:simpler-main}) of $\poly$ is only for ease of analysis; in reality, we will solve $\cppf$ over the original polytope given in (\ref{eq:main}) -- this is entirely equivalent to the above discussion.

%% file: 3.completion.tex
\section{Analysis of Weighted Completion Time: Theorem~\ref{thm:completion}}
	\label{sec:completion}

The analysis will be based on linear programming and dual fitting. Consider the following LP formulation, which is now standard for the weighted completion time objective~\cite{HallSSW97}.

\begin{align} 
	\min \sum_{t, j} w_j \cdot   \frac{t}{p_j} \cdot x_{jt}  \tag{$\primal$} \label{primal} \vspace{-3mm}
\end{align}
\vspace{-5mm}
\begin{align} 
	 &&  s.t. \quad \sum_{t \geq r_j}\frac{ x_{jt}}{p_j} &\geq 1&& \forall j \in U_0 \nonumber\\
	&&   B \cdot \vec{x}_t &\leq \vec{1} &&\forall t \geq 0\nonumber \\
	&& x_{jt} &\geq 0 &&\forall j,t \geq 0  \nonumber 
\end{align}

The variable $x_{jt}$ denotes the rate at which job $j$ is processed at time $t$. The first constraint ensures that each job must be completed. The second is the polytope constraint. It is easy to see that the objective lower bounds the actual total weighted flow time of any feasible schedule.

For a technical reason which will be clear soon, we will compare our algorithm to the optimal schedule with speed $1/s$, where $s$ will be set to $32$ later -- this is only for the sake of analysis, and the final result, as stated in Theorem~\ref{thm:completion}, will not need speed augmentation. The optimal solution with speed $1/s$ must satisfy the following LP. 
\begin{align} 
	\min \sum_{t, j} w_j \cdot   \frac{t}{p_j} \cdot x_{jt}   \tag{$\primal_s$} \label{primals}
\end{align}
\vspace{-5mm}
\begin{align} 
	 &&  s.t. \quad \sum_{t \geq r_j}\frac{ x_{jt}}{p_j} &\geq 1&& \forall j \in U_0 \nonumber\\
	&&   B \cdot (s\vec{x}_t) &\leq \vec{1} &&\forall t \geq 0\nonumber \\
	&&\quad   x_{jt} &\geq 0 && \forall j,t \geq 0  \nonumber
\end{align}

Note that the only change made in $\primal_s$ is that  $\xv$ is replaced with $s \xv$ in the second constraint.   We take the dual of this LP; here $\vec{\beta}_t := (\beta_{1t}, \beta_{2t}, ..., \beta_{Dt})$.

\begin{align}
	\max \sum_{j} \alpha_j -  \sum_{d,t}  \beta_{dt}  \tag{$\mathsf{DUAL_s}$} \label{dual} \vspace{-3mm}
\end{align}
\vspace{-5mm}
\begin{align}
	  &&s.t. \quad \frac{\alpha_j}{p_j} -  s B_{\cdot j} \cdot \betav_t &\leq w_j \cdot \frac{t}{p_j} &&\quad \forall j, t	\geq r_j \label{eqn:dual-1}\\
	&&	\quad				\alpha_j &\geq 0 			&&\quad \forall j  \label{eqn:dual-3}\\
	&&	\quad				\beta_{dt} &\geq 0 			&&\quad \forall d,t  \label{eqn:dual-4}
\end{align}

\medskip
We will set the dual variables $\alpha_j$ and $\beta_{dt}$ using the optimal solution of $\cppf$, $x^*_{jt}$, and the corresponding dual variables $y^*_{dt}$. The following proposition shows the outcome we will derive by dual fitting. 

\begin{proposition}
	\label{p:combine}
	Suppose there exist $\{\alpha_j\}_j$ and $\{\beta_{dt}\}_{d,  t}$ that satisfy all constraints in $\dual_s$ such that the objective of $\dual_s$ is at least $c$ times the total weighted completion time of algorithm $\cA$. Then $\cA$ is $(s/c)$-competitive for minimizing the total weighted completion time. 
\end{proposition}
\begin{proof}
	Observe that the optimal objective of $\primal_s$ is at most $s$ times that of $\primal$. This is because that any feasible solution $\vec{x}_t$ for $\primal$ is also feasible for $\primal_s$ when the $\vec{x}_t$ is stretched out horizontally by a factor of $s$ --  the new schedule $\vec{x}'_t$ is defined as $\vec{x}'_{(st)} = (1 / s)\vec{x}_{t}$ for all $t \geq 0$. The claim easily follows from the fact that $\primal$ is a valid LP relaxation of the problem,  weak duality, and the condition stated in the proposition.
\end{proof}

We will first show that the dual objective is a constant times the total weighted completion time of our algorithm, and then show that all dual constraints are satisfied.  Recall that $U_t:= \{ j \; | \; j < C^\cA_j\}$ denote the set of unsatisfied jobs at time $t$ -- it is important to note that $U_t$ also includes jobs that have not arrived by time $t$, hence could be different from the set $\cA_t := \{ j \; | \; r_j \leq j < C^\cA_j\}$ of alive jobs at time $t$. Let $W_t := \sum_{j \in U_t} w_j$ denote the total weight of unsatisfied jobs at time $t$.

\smallskip
We now show how to set dual variables using the optimal solution $\vec{x}^*_t$ of $\cppf$, and its dual variables $\vec{y}^*_t$. We will define $\alpha_{jt}$, and set $\alpha_j := \sum_{t} \alpha_{jt}$ for all $j$.  

 Let $q_{jt}$ denotes the size of  job $j$ processed at time $t$. Define $\zeta_{t}$ to be the `weighted' median of $\frac{q_{jt}}{p_j}$ amongst all jobs $j$ in $U_t$ -- that is, the median is taken assuming that each job $j$ in $U_t$ has $w_j$ copies. 
\begin{align*}
	\alpha_{jt} &:= 
		\begin{cases} 	 
			 w_j   	&\quad \forall j, t  \mbox{ s.t. } j \in U_t,  \frac{q_{jt}}{p_j} \leq \zeta_{t}  \\
			0     &\quad \mbox{otherwise}					 
		\end{cases} 	
\end{align*}

We continue to define $\beta_{dt}$ as $\beta_{dt} := \sum_{t' \geq t }  \frac{1}{s} \zeta_{t'} y^*_{dt'}$.  
We now show that this definition of $\alpha_{jt}$ and $\beta_{dt}$ makes $\dual_s$'s objective to be at least $O(1)$ times the objective of our algorithm. 

\begin{lemma}
	\label{lem:obj-alpha}
	$\sum_j \alpha_j \geq (1/2) \sum_{j} w_j C^\cA_j$.
\end{lemma}
\begin{proof}
	At each time $t$, jobs in $U_t$  contribute to $\sum_{j} \alpha_{jt}$ by at least half of  the total weight of jobs in $U_t$. 
\end{proof}

\begin{lemma}
	\label{lem:sumdual}
	For any time $t$, $\sum_{d} y^*_{dt} = \sum_{j \in \cA_t} w_j \leq W_t$.
\end{lemma}
\begin{proof}
	\begin{align*}
		\sum_{d} y^*_{dt}  =  \sum_d y^*_{dt} (B_{d\cdot} \cdot \xv^*_t) =  \sum_d y^*_{dt} \sum_{j \in \cA_t} B_{dj} \  x^*_{jt} = \sum_{j \in \cA_t}  x^*_{jt}  (B_{\cdot j} \cdot \vec{y}^*_{t})  = \sum_{j\in \cA_t} x^*_{jt}  \ \frac{w_j}{x^*_{jt}} 
		\leq W_t
	\end{align*}
	The first and last equalities are due to the KKT conditions (\ref{eqn:kkt-1}) and (\ref{eqn:kkt-2}), respectively. 
\end{proof}

\begin{lemma}
	\label{lem:obj-beta}
	At all times $t$, $\sum_{d}  \beta_{dt} \leq \frac{8}{s} W_t$.
\end{lemma}
\begin{proof}
	Consider any fixed time $t$.  We partition the time interval $[t, \infty)$ into subintervals  $\{M_k\}_{k \geq 1}$ such that the total weight of unsatisfied jobs at all times during in $M_k$ lies in the range $\Big((\frac{1}{2})^{k} W_t, (\frac{1}{2})^{k-1} W_t \Big]$. Now consider any fixed $k \geq 1$. We upper bound the contribution of $M_k$ to $\sum_{d} \beta_{dt}$, that is $\frac{1}{s} \sum_{t'  \in M_k} \sum_d \zeta_{t'}  y^*_{dt'}$. Towards this end, we first upper bound $\sum_{t' \in M_k} \zeta_{t'} \leq 4$. The key idea is to focus on the total weighted throughput processed during $M_k$. Job $j$'s fractional weighted throughput at time $t'$ is defined as $w_j\frac{q_{jt'}}{p_j}$, which is job $j$'s weight times the fraction of job $j$ that is processed at time $t'$; recall that $q_{jt'}$ denotes the size of job $j$ processed at time $t'$.

\begin{align*}
	\sum_{t' \in M_k} \zeta_{t'}
	&\leq \sum_{t' \in M_k} 2 \sum_{j \in \cA_{t'}} \frac{w_j}{W_{t'}}   \cdot \oneindi \Big( \frac{q_{jt'}}{p_{j}} \geq \zeta_{t'} \Big)  \cdot  \frac{q_{jt'}}{p_j} \leq 2 \frac{1}{(1/2)^k W_t} \sum_{t' \in M_k}   \sum_{j \in U_{t'}}  w_{j}  \ \frac{q_{jt'}}{p_j} \\
	&\leq 2 \frac{1}{(1/2)^k W_t}   (1/2)^{k-1} W_t 
	=4
\end{align*}
	
	The first inequality follows from the definition of $\zeta_{t'}$: for jobs $j$ with total weight at least half the total weight of jobs in $U_{t'}$, $\frac{q_{jt'}}{p_j} \geq \zeta_{t'}$. The second inequality is due to the fact that $W_{t'} \geq (\frac{1}{2})^{k} W_t$ for all times $t' \in M_k$. The last inequality follows since the total weighted throughput that can be processed during $M_k$ is upper bounded by the weight of unsatisfied jobs at the beginning of $M_{k'}$, which is at most $(\frac{1}{2})^{k-1} W_t$. Therefore,
\begin{align*}	
	 \sum_d  \beta_{dt} 
	&= \frac{1}{s} \sum_{t' \geq t} \sum_d \zeta_{t'} y_{dt'}^* 
	= \frac{1}{s} \sum_{k \geq 1} \sum_{t' \in M_k} \zeta_{t'} \sum_{d}  y^*_{dt'} \\
	&\leq \frac{1}{s} \sum_{k \geq 1} \sum_{t' \in M_k} \zeta_{t'} W_{t'} &&\mbox{ [By Lemma~\ref{lem:sumdual}]}\\
	&= \frac{1}{s} \sum_{k \geq 1} 4 (1/2)^{k-1} W_{t} &&\mbox{ [By definition of $M_{k}$ and the fact $\sum_{t' \in M_k} \zeta_{t'} \leq 4$]}\\
	&\leq \frac{8}{s}  W_t		
\end{align*}	
\end{proof}

\begin{corollary}
	\label{coro:obj-beta}
	$\sum_{d,t}  \beta_{dt} \leq \frac{8}{s} \sum_j w_j C^\cA_j$. 
\end{corollary}

From Lemma~\ref{lem:obj-alpha} and Corollary~\ref{coro:obj-beta}, we derive that the objective of $\dual_s$ is at least half of $\pf$' total weighted completion time for $s  = 32$. By Lemma~\ref{p:combine}, it follows that the algorithm $\pf$ is $64$-competitive for the objective of minimizing the total weighted completion time. 

It now remains to show all the dual constraints are satisfied.  Observe that the dual constraint (\ref{eqn:dual-3}) is trivially satisfied. Also the constraint (\ref{eqn:dual-4}) is satisfied due to KKT condition (\ref{eqn:kkt-3}). 	

We now focus on the more interesting dual constraint (\ref{eqn:dual-1}) to complete the analysis of Theorem~\ref{thm:completion}.

\begin{lemma}
	The dual constraint (\ref{eqn:dual-1}) is satisfied. 
\end{lemma}
\begin{proof}
\begin{align*}
	\frac{\alpha_j}{p_j} -  w_j \frac{t}{p_j} 
	&\leq \sum_{t' \geq t} \frac{\alpha_{jt'}}{p_j}   && \mbox{[Since $\alpha_{jt'} \leq w_j$ for all $t'$]}\\	
	&= \sum_{t' \geq t} \frac{ w_j}{p_j}\cdot \oneindi \Big( \frac{q_{jt'}}{p_j} \leq \zeta_{t'}\Big) = \sum_{t' \geq t} \frac{w_j}{q_{jt'}} \cdot \frac{q_{jt'}}{p_j}  \cdot \oneindi \Big( \frac{q_{jt'}}{p_j} \leq  \zeta_{t'}\Big) &&\\
	&= \sum_{t' \geq t} \frac{w_j}{x^*_{jt'}} \cdot \frac{q_{jt'}}{p_j}  \cdot \oneindi \Big( \frac{q_{jt'}}{p_j} \leq \zeta_{t'}\Big)     && \mbox{[Since $q_{jt'}  = x^*_{jt'}$]}\\	
	&\leq \sum_{t' \geq t} B_{\cdot j} \cdot (\zeta_{t'} \yv^*_{t'})     && \mbox{[By the KKT condition (\ref{eqn:kkt-2})]}\\	
	&= s B_{\cdot j} \cdot \betav_t    && \mbox{[By definition of $\betav_t$]}
\end{align*}
\end{proof}

%% file: 5.unrelated.tex
\section{Unrelated Machine Scheduling and Theorem~\ref{thm:flow-no-scalable}}
\label{sec:unrelated}
We now consider the canonical special case of \psp termed {\em unrelated machine scheduling}. In this problem, there are $M$ unrelated machines. Job $j$ is processed at rate $s_{ij}  \in [0, \infty)$ on each machine $i$. (Unrelated machines generalizes related machines where machines have different speeds independent of jobs). The online algorithm is allowed to preempt and migrate jobs at any time with no penalty. The important constraint is that at any instantaneous time, 
a job can be processed only on a single machine. The objective is to minimize the sum of unweighted flow times $\sum_j F_j$.

In the non-clairvoyant scheduling model we consider, the scheduler knows the values of $s_{ij}$ but is not aware of the processing length $p_j$. It is easy to show that if $s_{ij}$ values are not known as well, then no online algorithm can have a bounded competitive ratio even with $\Omega(1)$-speed augmentation. Our main result is a {\em scalable} algorithm that,  for any $\epsilon > 0$, is $O(1/\epsilon^2)$ competitive with speed augmentation of $(1+\epsilon)$.

\subsection{The Balanced Latest Arrival Smooth Scheduling (BLASS) Algorithm}
The BLASS algorithm has three main components. 
\begin{description}
\item[Single-machine Scheduling Policy:] This determines the rates at which jobs assigned to a single machine are processed. In our algorithm, each machine runs a new scheduling policy called {\em Smoothed Latest Arrival Processor Sharing} (\slaps). \slaps is preemptive and non-clairvoyant.
\item[Dispatch Policy:] The dispatch policy determines the machine to which an arriving job is assigned. 
\item[Rearrangement Policy: ] This procedure is called upon the completion of a job. The procedure changes the assignment of jobs to machines. This is the crux of the algorithm.
\end{description}

\noindent We describe each component of the algorithm in more detail.

\subsubsection{Single-machine Scheduling Policy \slaps$(k)$}
	\label{sec:algorithm}
The policy $\slaps(k)$ is parameterized by a non-negative integer $k$. We will set $k = 1/\epsilon$, where $\epsilon$ is the extra speed given to the online algorithm; for simplicity, we assume that $1/ \eps$ is an integer. Focus on a particular machine, say machine $i$, and let $J_i(t)$ denote the set of alive jobs at time $t$ that are assigned to this machine and $N_i(t) = |J_i(t)|$. Order these jobs in increasing order of their arrival time. Let $\pi_j(t)$ denote job $j$'s rank at time $t$ in this order. Here, the earliest arriving job has rank 1 and the latest arriving job has rank $N_i(t)$. We assume w.l.o.g. that all jobs have distinct arrival times by breaking ties arbitrarily but consistently. We distribute the total processing among the jobs in $J_i(t)$ by giving each job $j$ the following fraction of its rate:

\begin{equation}
	\label{eqn:jobshare}
	\nu_j(t)  :=  \frac{ {\pi_j(t)}^k }{ 1^k + 2^k + ... + (N_i(t))^k  }, 
\end{equation}
\noindent

Therefore, the rate of processing job $j$ at time $t$ is $\nu_j(t) \times s_{ij}$. This completes the description of the algorithm that each machine runs. It is worth noting that $\slaps(k)$ with $k=0$ is exactly \textsc{Round Robin} ($\rr$). As $k$ grows, $\slaps(k)$ shifts its processing power towards more recently arriving jobs and the shift is made smooth as alluded by the name of the scheduling algorithm. However, $\slaps(k)$ retains many important characteristics of $\rr$. One crucial property of $\rr$ is that a pair of jobs delay each other by exactly the same amount. $\slaps(k)$ preserves this property to a factor of $O(k)$. 

The following proposition has an elementary proof.

\begin{proposition}
	\label{prop:upper}
	Consider any integer $k \geq 0$, and $n \geq 1$. Then
	$$\frac{n^k}{1^k + 2^k + ... + n^k} \leq \frac{k+1}{n} \qquad \mbox{and} \qquad \frac{n^k}{1^k + 2^k + ... + (n-1)^k} \geq \frac{k+1}{n}$$
\end{proposition}

\subsubsection{The Dispatch and Rearrangement Policies}
\begin{definition}
The {\em global-rank} of a job is the index of the job in the sorted list of jobs $J$, sorted in ascending order by release dates: $\Pi_j = |\{j' \; |\; r_{j'} < r_j\}| + 1$. 
\end{definition}

For any job $j$ and time $t$, let $\sigma(j, t)$ denote the machine to which the job is assigned by the BLASS algorithm.  For any machine $i$, job $j$, and time $t$, let $J_{< j}(i,t)$ denote the set of jobs that have global-rank less than global-rank of job $j$  and that are assigned to machine $i$ at time $t$: $J_{< j}(i,t) = \{j' \;|\; r_{j'} < r_j, \sigma(j', t) = i \}$.  Let $N_{< j}(i,t) = |J_{< j}(i,t)|$. Similarly, define $N_{\le j}(\sigma(j,t),t) = N_{< j}(\sigma(j,t),t) + 1$.

\begin{definition}
The {\em local-rank} of a job $j$ is the rank of the job with respect to the set of jobs which are assigned to machine $i$ at time $t$: $\pi_j(i, t) = N_{<j}(i,t) + 1$. 
\end{definition}

In the BLASS algorithm, each machine runs $\slaps(k)$ on the set of jobs assigned to it. For a job $j$ assigned to machine $i$, let $\nu_{j}(i,t)$ denote the fraction of rate $s_{ij}$ that job $j$ receives in this policy. If $N_i(t)$ is the number of jobs assigned to machine $i$ at time $t$, we have from equation (\ref{eqn:jobshare}), 

\begin{equation}
\label{eqn:speed}
\nu_{j}(i,t)  := \frac{ {\pi_{j}(i,t)}^k }{ 1^k + 2^k + ... + (N_i(t))^k  }
\end{equation}

 Let $L(i,j,t) = \frac{s_{ij}}{N_{< j}(i,t) + 1}$. This is the hypothetical rate at which job $j$ will be processed if it were assigned to machine $i$ and this machine were to execute $\rr$ on jobs with local rank at most that of job $j$.

\medskip
\noindent \textbf{The Dispatch Policy:} When a job $j$ arrives, it is assigned to the machine $i$ which maximizes $L(i,j,t)$.

\medskip
\noindent \textbf{Rearrangement Policy:} Suppose a job $j^*$ completes on machine $i^*$ at time $t$. Then, the algorithm calls the procedure {\sc Rearrange} described in Figure~\ref{fig1}. The procedure maintains a machine $b$ on which there is {\em slack}. Initially this is $i^*$, since a job departed from here. It considers jobs in increasing order of global rank starting at $j^*$ and finds the first job $j$ for which $ L(b,j,t) > L(\sigma(j,t),j,t)$. The job $j$ is moved from $i = \sigma(j,t)$ to machine $b = i^*$. This frees up processing on machine $i$, so that $b \leftarrow i$ now. Note that the quantities $L$ can change as a result. It iterates on the remaining jobs.

\begin{figure}[htbp]
\begin{center}
\fbox{\begin{minipage}{6in}
\begin{tabbing}
{\sc Rearrange} \\
\ \ \= /* Order the alive jobs in increasing order of $\Pi_j$ and denote them $1,2,\ldots, l$ in this order. */\\
\> /* Let $j^*$ denote the job that has departed from machine $i^*$. */\\ \\
\> $b \leftarrow i^*$.   \qquad \qquad /* $b$ is the current machine with {\em slack}. */ \\
\> {\bf For} $j := j^* + 1$ {\bf to} $l$ {\bf do}: \\
\> \ \ \= (Implicitly) recompute the quantities $L$. \\
\> \> {\bf If} $L(b,j,t) > L(\sigma(j,t),j,t)$ {\bf then}\\
\> \> \ \ \= $i \leftarrow \sigma(j,t)$ \\
\> \> \> Assign the job $j$ to machine $b$, {\em i.e.} set $ \sigma(j,t) \leftarrow b$ \\
\> \> \> $b \leftarrow i$. \\
\> \> {\bf EndIf} \\
\> {\bf EndFor}
\end{tabbing}
\end{minipage}}
\end{center}
\caption{\label{fig1}The {\sc Rearrange} Procedure.}
\end{figure}

We now show several monotonicity properties of $L(i,j,t)$. In particular, we show that any time instant, jobs are assigned to machines for which this quantity is largest (fixing the assignment of other jobs), and furthermore, this quantity is non-decreasing with time for any job. These properties will be crucial in our analysis, and their proof hinges on considering jobs in increasing order of global rank in {\sc Rearrange}.

\begin{lemma}
\label{lem:monotonicity}
Let $[t_1, t_2]$ be any interval of time where a job $j$ is on the machine $i$ for all time instants $t \in [t_1, t_2]$. Then, $N_{<j}(i,t)$ is non-increasing in $t$.
\end{lemma}
\begin{proof}
Fix a job $j$ assigned to machine $i$. The set of jobs on machine $i$ can change only upon the arrival of a job or on the completion of a job either on machine $i$ or some other machine $i'$. The arrival of a job does not change the quantity $N_{<j}(i,t)$ since the global-rank of an arriving job is greater than the global-rank of job $j$. When {\sc Rearrange} is invoked, $N_{<j}(i,t)$ increases if a job $j'$ with smaller global rank than $j$ moves into machine $i$. But this must be preceded by a job with even smaller rank moving {\em out} of $i$. Therefore, $N_{<j}(i,t)$ is non-increasing for all time $t \in [t_1,t_2]$.
\end{proof}

Next we show that at all time instants, each job $j$ is assigned to that machine which locally optimizes $L$.

\begin{lemma}
\label{lem:best}
For all jobs $j$, machines $i$, and time instants $t$ where $j$ is alive, $L(\sigma(j,t),j,t) \geq L(i,j,t)$.
\end{lemma}
\begin{proof}
We prove the lemma by induction on $t$. Consider a job $j$. From the dispatch rule, the lemma is true at $t = r_j$.  Suppose the condition is true at time $t_1 > r_j$.  The assignment of jobs to machines changes only upon the arrival or completion of a job. The arrival of a job does not change the invariant since the global-rank of the arriving job is greater than that of $j$. During {\sc Rearrange}, the quantity $L(i,j,t)$ can only change if a job with lower global rank than $j$ departs. In the iteration of the for loop, as long as $j$ is not reached, the machine $b$ loses one job and gains one job, so that $L(b,j,t)$ remains unchanged. (The $L(i,j,t)$ for remaining $i$ remain trivially unchanged.) When the loop reaches job $j$, either job $j$ moves to $b$ if $L(b,j,t) > L(\sigma(j,t),j,t)$ or it stays where it is. Subsequently, only jobs with global rank greater than that of $j$ change machines, and this cannot affect $L(i,j,t)$ for any $i$. This completes the proof.
\end{proof}

\begin{lemma}
\label{lem:increase}
For all jobs $j$, $L(\sigma(j,t),j,t)$ is non-decreasing in $t$ over the time interval where $j$ is alive.
\end{lemma}
\begin{proof}
By Lemma~\ref{lem:monotonicity}, if a job does not change machines in a time interval, this condition is true. By Lemma~\ref{lem:best}, this is also true when {\sc Rearrange} is invoked, completing the proof.
\end{proof}

\subsection{Analysis of BLASS: Proof of Theorem~\ref{thm:flow-no-scalable}}
We first write a linear programming relaxation of the problem \lpp described below which was first given by~\cite{AnandGK12,GargK07}. It has a variable $x_{ijt}$ for each machine $i \in [m]$, each job $j \in [n]$ and each unit time slot $t \geq r_{j}$. If the machine $i$ processes the job $j$ during the whole time slot $t$, then this variable is set to $1$. The first constraint says that every job has to be completely processed. The second constraint says that a machine cannot process more than one unit of jobs during any time slot. Note that the LP allows a job to be processed simultaneously across different machines.

\[ \mbox{Min} \ \ \ \sum_{j} \sum_{i} \sum_{t \geq r_{j}}  \left(\frac{ s_{ij}(t- r_j)}{p_{j}} + 1 \right) \cdot x_{ijt}  \qquad \tag{\lpp} \label{primal2} \]
\[ \begin{array}{rcllr}
\displaystyle 	\sum_{i} \sum_{t \geq r_{j}} \frac{s_{ij} \cdot x_{ijt}}{p_{j}} &\geq& 1 \qquad &\forall j   \\ 
\displaystyle	\sum_{j \, : \, t \geq r_{j}} x_{ijt} &\leq& 1 &\forall i,    t      \\
\displaystyle	x_{ijt} &\geq& 0  &\forall i,  j,  t \, : \, t \geq r_{j} \qquad 
\end{array}  \]

It is easy to show that the above LP lower bounds the optimal flow time of a feasible schedule within factor 2. As before, we use the dual fitting framework. We write the dual of $\lpp$ as follows.

\[ \mbox{Max} \ \ \ \sum_{j} \alpha_j -   \sum_{i} \sum_t \beta_{it} \qquad \tag{\lpd} \label{Dual} \]
\[ \begin{array}{rcllr}
\displaystyle \;\;\;\;	\frac{s_{ij} \cdot \alpha_j}{p_{j}} - \beta_{it} &\leq& \displaystyle \frac{s_{ij}(t- r_j)}{p_{j}} + 1 \qquad &\forall  i,j, t \, : \, t \geq r_{j}  \label{dual-constraint}\\ 
\displaystyle  \;\;\;\;	\alpha_j &\geq& 0  &\forall   j \\ 
\displaystyle  \;\;\;\;	\beta_{it} &\geq& 0  &\forall   i,t  
\end{array}  \]

We will show that there is a feasible solution to $\lpd$ that has objective $O(\epsilon^2)$ times the total flow time of BLASS, provided we augment the speed of BLASS by $\eta = (1 + 3 \epsilon)$. From now on, we will assume that each processor in BLASS has this extra speed when processing jobs. This modifies Equation (\ref{eqn:speed}) to:
$$ \nu_{j}(i,t)  := \eta \times \frac{ {\pi_{j}(i,t)}^k }{ 1^k + 2^k + ... + (N_i(t))^k  }$$

 Recall that each machine runs $\slaps$ with $k = 1/\epsilon$, and we assume without loss of generality that $1/\epsilon$ is an integer. Let $d_{jj'}$ denote the delay caused by job $j$ to job $j'$. The quantity $d_{jj'}$ is equal to the total processing done on job $j$ in the intervals where $j$ and $j'$ are on the same machine, scaled down by $\eta$. Let $\oneindi_{\{\sigma(j,t) = \sigma(j',t)\}}$ be the indicator function denoting whether job $j$ and $j'$ are assigned to the same machine at time $t$. Then,
$$
d_{jj'} = \frac{1}{\eta}\int^{C_j}_{r_j} \nu_{j}(\sigma(j,t),t) \cdot \oneindi_{\{\sigma(j,t) = \sigma(j',t)\}}dt
$$
Similarly, let $p_j(t^*) = \int^{C_j}_{t = t^*} s_{\sigma(j,t)j} \cdot \nu_{j}(\sigma(j,t),t) dt$ denote the residual size of job $j$ at time $t^*$.

Furthermore, let  $d_{jj'}(t,C_j) =  \frac{1}{\eta} \int^{C_j}_{t} \nu_{j}(\sigma(j,t),t) \oneindi_{\{\sigma(j,t) = \sigma(j',t)\}}dt$ denote the total delay incurred by job $j'$ due to job $j$ after time instant $t$.  Let $J_{<j}$ denote the set of jobs which have release dates less than release date of job $j$. For any job $j$ define,

$$\Delta_j = \sum_{\{j' \in J_{< j} \cup j\} } d_{jj'} +   \sum_{\{j' \in J_{<j}\}} d_{j'j} $$ 

 $\Delta_j$ denotes the total delay the job $j$ causes on all jobs in the set $J_{<j}$ plus the delay caused by the jobs in set $J_{<j}$ on the job $j$. Note that in the first term in the summation, we are including the total processing done on job $j$ itself. Similarly, define $\Delta_j(t,C_j) = \sum_{\{j' \in J_{<j} \cup j\}} d_{jj'}(t,C_j) +   \sum_{\{j' \in J_{<j}\}} d_{j'j}(t,C_j) $ as the residual delay after time $t$. Note that $\sum_{j} \Delta_j = \sum_{j} F_j$, where $F_j$ denotes the flow time of job $j$ in BLASS. 

Next, we establish the following important upper bound on the quantity $\Delta_j(t,C_j)$. 
\begin{lemma}
\label{lem1}
$$ \Delta_j(t^*,C_j) \le  \frac{1}{\eta} \cdot\frac{k+2}{k+1} \cdot \frac{p_j(t^*)}{L(\sigma(j,t^*),j,t^*)}$$
\end{lemma}
\begin{proof}
{\small
\begin{eqnarray*}
\Delta_j(t^*,C_j) &=& \frac{1}{\eta} \int^{C_j}_{t = t^*} \left( \nu_{j}(\sigma(j,t),t) \cdot N_{\leq j}(\sigma(j,t),t) + \displaystyle \sum_{j' \in J_{<j}(\sigma(j,t),t)} \nu_{j'}(\sigma(j,t),t) \right) dt   \nonumber \\
&=& \frac{1}{\eta}  \int^{C_j}_{t = t^*} \left( \nu_{j}(\sigma(j,t),t) \cdot  N_{\leq j}(\sigma(j,t),t) + \frac{\eta \cdot \left(\displaystyle \sum_{j' \in J_{<j}(\sigma(j,t),t)}{\pi_{j'}(\sigma(j,t),t)}^k \right)}{\displaystyle \sum^{N_i(t)}_{a = 1}a^k} \right) dt  \nonumber\\ \nonumber \\
&=& \frac{1}{\eta}  \int^{C_j}_{t = t^*} \left( \nu_{j}(\sigma(j,t),t) \cdot  N_{\leq j}(\sigma(j,t),t) + \left( \frac{\displaystyle \sum_{j' \in J_{<j}(\sigma(j,t),t)}{\pi_{j'}(\sigma(j,t),t)}^k}{{\pi_j(\sigma(j,t),t)}^k} \cdot \frac{\eta \cdot {\pi_j(\sigma(j,t),t)}^k}{\displaystyle \sum^{N_i(t)}_{a = 1}a^k} \right) \right)dt  \nonumber \\
&=& \frac{1}{\eta}  \int^{C_j}_{t = t^*} \nu_{j}(\sigma(j,t),t) \cdot \left( N_{\leq j}(\sigma(j,t),t) + \frac{\displaystyle \sum_{j' \in J_{<j}(\sigma(j,t),t)}{\pi_{j'}(\sigma(j,t),t)}^k }{{\pi_j(\sigma(j,t),t)}^k} \right) dt  \nonumber \\
&\leq& \frac{1}{\eta}  \int^{C_j}_{t = t^*} \nu_{j}(\sigma(j,t),t) \cdot \left( N_{\leq j}(\sigma(j,t),t) + \frac{N_{\leq j}(\sigma(j,t),t)}{k+1} \right) dt  \qquad \mbox{[Proposition~\ref{prop:upper}]}\nonumber \\
&=&  \frac{1}{\eta} \cdot\frac{k+2}{k+1} \cdot \int^{C_j}_{t = t^*} \nu_{j}(\sigma(j,t),t) \cdot  N_{\leq j}(\sigma(j,t),t) dt \nonumber \\
&=&  \frac{1}{\eta} \cdot \frac{k+2}{k+1} \cdot \int^{C_j}_{t = t^*} s_{\sigma(j,t)j}\cdot \nu_{j}(\sigma(j,t),t)  \cdot \frac{1}{L(\sigma(j,t),j,t)} dt \nonumber \\
&\leq&  \frac{1}{\eta} \cdot \frac{k+2}{k+1} \cdot \frac{1}{L(\sigma(j,t^*),j,t^*)} \cdot \int^{C_j}_{t = t^*} s_{\sigma(j,t)j}\cdot \nu_{j}(\sigma(j,t),t) dt \qquad \mbox{[Lemma~\ref{lem:increase}]}\nonumber \\
&=&  \frac{1}{\eta} \cdot\frac{k+2}{k+1} \cdot \frac{p_j(t^*)}{L(\sigma(j,t^*),j,t^*)}
\end{eqnarray*}}
\end{proof}

Using the above inequality, it is easy to show the following:

\begin{lemma}
\label{lem:delay}
For any time instant $t^*$ and for any job $j$, 
$$
\Delta_j \leq \quad(k+2) (t^*-r_j) + \left( \frac{1}{\eta} \cdot \frac{k+2}{k+1} \cdot \frac{p_{j}(t^*)}{L(\sigma(j,t^*),j,t^*)} \right)
$$
\end{lemma}
\begin{proof}
{\small
\begin{eqnarray*}
\Delta_j &=& \frac{1}{\eta} \int^{t^*}_{t = r_j} \left( \nu_{j}(\sigma(j,t),t) \cdot N_{\leq j}(\sigma(j,t),j) + \displaystyle \sum_{j' \in J_{<j}(\sigma(j,t),t)} \nu_{j'}(\sigma(j,t),t) \right) dt   + \Delta_j(t^*, C_j) \nonumber \\
&\leq& \frac{1}{\eta} \int^{t^*}_{t = r_j} \left( \frac{ \eta \cdot {\pi_j(\sigma(j,t),t)}^k}{\displaystyle \sum^{N_{\leq j}(\sigma(j,t),t)}_{a = 1}a^k} \cdot N_{\leq j}(\sigma(j,t),j) + \eta \right) dt   + \Delta_j(t^*, C_j) \qquad \mbox{[From the def of $\nu$, Equation~\ref{eqn:speed}]}\nonumber\\
&\leq& \frac{1}{\eta} \int^{t^*}_{t = r_j} \left( \frac{\eta ( k+1)}{N_{\leq j}(\sigma(j,t),t)} \cdot N_{\leq j}(\sigma(j,t),t) + \eta \right) dt  + \Delta_j(t^*, C_j) \qquad \mbox{[Proposition~\ref{prop:upper}]} \nonumber \\
&\leq&  \int^{t^*}_{t = r_j}(k+2) dt  + \Delta_j(t^*,C_j) \nonumber \\
&\leq& (k+2)(t^*-r_j) + \Delta_j(t^*,C_j)
\end{eqnarray*}}
Combined with Lemma~\ref{lem1}, this completes the proof.
\end{proof}

We now perform the dual fitting. We set the variables of the \ref{Dual} as follows. We set $\beta_{it}$ proportional to the total number of jobs alive on machine $i$ at time $t$: $\beta_{it} = \frac{1}{k+3} N_i(t)$. We set $\alpha_j$ proportional to the total delay  the job $j$ causes to the jobs of global-rank at most $\Pi_j$: $\alpha_j = \frac{1}{k+2} \Delta_j$.

We first bound the dual objective as follows (noting $k = 1/\epsilon$ and $\eta = 1+3 \epsilon$):
\begin{eqnarray}
\sum_{j} \alpha_j - \sum_{i,t}\beta_{it} &=&  \sum_{j} \frac{\Delta_j}{k+2} - \sum_{i,t} \frac{N_i(t)}{k+3} =  \epsilon \left(\sum_{j} \frac{\Delta_j}{1+2\epsilon} - \sum_{i,t} \frac{N_i(t)}{1 + 3\epsilon} \right) \nonumber\\ 
&=& \epsilon \cdot \sum_j F_j \cdot \left(\frac{1}{1+2\epsilon} - \frac{1}{1+3\epsilon}\right) = O(\epsilon^2) \sum_j F_j
\end{eqnarray}

It therefore remains to prove that  constraints of \ref{Dual} are satisfied. To see this, fix job $j$ and time instant $t$. We consider two cases.

\noindent \textbf{Case 1:} Machine $i = \sigma(j,t)$. Then
\begin{eqnarray}
\alpha_j - \frac{p_j}{s_{ij}}\beta_{it} &=& \frac{\Delta_j}{k+2} - \frac{p_j}{s_{ij}}\cdot \frac{N_i(t)}{k+3}  \nonumber \\
&\leq& (t-r_j) + \frac{p_j(t)}{\eta \cdot (k+1)} \cdot \frac{N_{<j}(i,t)+1}{s_{ij}} - \frac{p_j}{s_{ij}}\cdot \frac{N_i(t)}{k+3}  \qquad \mbox{[Lemma~\ref{lem:delay}]} \nonumber \\
&\leq& t - r_j \qquad [\mbox{since } \eta = 1+3\epsilon,k = 1/\epsilon]\nonumber
\end{eqnarray}

\noindent \textbf{Case 2:} Machine $i \neq \sigma(j,t)$. Then
\begin{eqnarray}
\alpha_j - \frac{p_j}{s_{ij}}\beta_{it} &=& \frac{\Delta_j}{k+2} - \frac{p_j}{s_{ij}}\cdot \frac{N_i(t)}{k+3} \nonumber \\
&\leq& t-r_j + \frac{1}{\eta} \cdot \frac{p_j(t)}{k+1}\cdot \frac{N_{<j}(\sigma(j,t), t) + 1}{s_{\sigma(j,t)j}} - \frac{p_j}{s_{ij}}\cdot \frac{N_i(t)}{k+3} \qquad \mbox{[Lemma~\ref{lem:delay}]} \nonumber \\
&\leq& t-r_j + \left (\frac{p_j(t)}{\eta\cdot(k+1)}\frac{N_{<j}(i,t)+1}{s_{ij}} - \frac{p_j}{k+3} \cdot \frac{N_i(t)}{s_{ij}} \right) \qquad \mbox{[Lemma~\ref{lem:best}]} \nonumber \\
&\leq& t-r_j + \frac{p_j}{ s_{ij}} \qquad [\mbox{since } \eta = 1+3\epsilon, k = 1/\epsilon]\nonumber \nonumber
\end{eqnarray}

Therefore, the dual constraints are satisfied for all time instants $t$ and all jobs $j$, and we derive that our algorithm is $(1+\eps)$-speed 
$O(1/\epsilon^2)$-competitive against $\lpp$, completing the proof of Theorem \ref{thm:flow-no-scalable}.

%% file: 6.conclusion.tex
\section{Conclusions and Open Questions}
We mention several open questions. Our results for weighted flow time of \psp, though tight, are weak due to the generality of the problem we consider. Our worst-case examples, however, are quite artificial, and we believe that the each individual application (such as unrelated machine scheduling or multi-dimensional scheduling) will admit stronger positive results. We have indeed shown such a result for unweighted flow time on unrelated machines. We believe that there could exist a $O(1)$-competitive algorithm for multidimensional scheduling with $O(1)$- or $O(\log M)$-speed (the speed is independent of the total number of jobs arriving at the system).  In the same vein, it is open whether there is a $O(1)$-speed $O(1)$-competitive (clairvoyant) algorithm for $\psp$ for minimizing the total weighted flow time. We hope our positive results for the general $\psp$ will lead to a more systematic study of scheduling problems that arise in modern data center applications. 

We note that Im \etal recently obtained a non-clairvoyant scalable algorithm on unrelated machines for minimizing the total weighted flow time \cite{ImKMP14}. Further, their algorithm  extends to the objective of minimizing the total weighted flow time plus energy consumption.

%% file: 4.flowtime-ub.tex
\section{Total Weighted Flow Time of Proportional Fairness: Theorem~\ref{thm:flow}}
	\label{sec:flowtime}

In this section we study the more challenging objective of minimizing the total weighted flow time for \psp. As discussed before, resource augmentation will be needed to obtain positive results for the flow objective.  We prove the upper and lower bound results claimed in Theorem~\ref{thm:flow} in Section~\ref{sec:flow-upper} and  \ref{sec:flow-lower}, respectively. As before, we use the same algorithm \pf and denote it by $\cA$.

\subsection{Upper Bound}
	\label{sec:flow-upper}

\subsubsection{High-level Overview} 

We first give a high-level overview of the analysis. Similar to the analysis we did for the completion time objective, we will use dual fitting. The mathematical programs should be modified accordingly. The only changes made in $\primal$ and $\primal_s$ are in the objective: $\sum_{t, j \in U(t)} w_j \cdot   \frac{t}{p_j} \cdot x_{jt}$ should be changed to $\sum_{t \in r_j, t} w_j \cdot   \frac{t - r_j}{p_j} \cdot x_{jt}$. Recall that $\primal_s$ gives a valid lower bound to the optimal adversarial scheduler with $1/s$-speed, which is equivalent to the algorithm being given $s$-speed and the optimal scheduler being given 1-speed. 
Then the dual of $\primal_s$ is as follows. 
\begin{align}
	\max \sum_{j} \alpha_j -  \sum_{d,t}  \beta_{dt}  \tag{$\mathsf{DUAL_s}$} \label{f-dual} 
\end{align}
\vspace{-5mm}
\begin{align}
	  &&s.t. \quad \frac{\alpha_j}{p_j} -  s \Bbf_{\cdot j} \cdot \betav_t &\leq w_j \cdot \frac{t - r_j}{p_j} &&\quad \forall j, t	\geq r_j \label{eqn:dual-1-flow}\\
	&&	\quad				\alpha_j &\geq 0 			&&\quad \forall j  \label{eqn:dual-3-flow}\\
	&&	\quad				\beta_{dt} &\geq 0 			&&\quad \forall d,t  \label{eqn:dual-4-flow}
\end{align}

Let's recall the high-level idea of the analysis for the completion time objective. In dual fitting, we set dual variables so that (i) the dual objective is at least $O(1)$ times $\cA$' total weighted completion time of PF, and (ii) all dual constraints are satisfied. Regarding (i), we had two things to keep in mind: to  make $\sum_{j} \alpha_j$ comparable to $\cA$'s total weighted completion time, and make  $\sum_{d,t}  \beta_{dt}$ smaller than $\sum_{j} \alpha_j$. On the other hand, in order to satisfy dual constraints, we have to ``cover''  the quantity $\frac{\alpha_j}{p_j}  - w_j \cdot \frac{t}{p_j}$ by $\beta_{dt}$. In this sense, (i) and (ii) compete each other. 

To give an overview of the analysis of the flow result, for simplicity, let's assume that all jobs are unweighted. We will construct a laminar family $\cL$ of intervals, and fit jobs into intervals with similar sizes. By ignoring jobs of small flow time, say less than $\frac{1}{2n}$ times the maximum weighted flow time of all jobs, we can assume that there are at most $O(\log n)$ levels of intervals in the family. Now we find the level such that the total weighted flow time  of jobs in the level is maximized. We crucially use the  fact that all jobs in the level have similar flow times. Namely, we will be able to define $\beta_{dt}$ for each laminar interval in the level we chose, via a linear combination of $\vec{y}_t$ over  the laminar interval. 

The actual analysis is considerably  more subtle particularly due to jobs with varying weights. In the weighted case, we will have to create multiple laminar families. Further, there could be many different laminar families, even more than $polylog(n)$. It is quite challenging to define $\beta_{dt}$ using a linear combination of the duals of $\cppf$ while trying to minimize the side-effect between different laminar families. Also ``ignoring" jobs is not as easy as it looks since such ignored jobs still contribute the duals of $\cppf$.

\subsubsection{Main Analysis}

In this section we give a formal analysis of the upper bound claimed in Theorem~\ref{thm:flow}. To make our analysis more transparent, we do not optimize constants. To set dual variables, we perform the following sequence of preprocessing steps. To keep track of changes made in each step, we define a set $\cA'_t \subseteq \cA_t$ of ``active" jobs at time $t$, which will help us to set dual variables later. Also we will maintain the set of ``globaly active" jobs $\cA'$. The set $\cA'$ is global in the sense that if $j \notin \cA'$ then $j \notin \cA'_t$ for all $t$. Intuitively, jobs in $\cA'$ will account for a large fraction ($\Omega( 1/  \log n)$) of the total weighted flow time. Initially, we set $\cA'_t := \cA_t$ for all $t$ (recall that $\cA_t$ denote jobs alive at time $t$), and $\cA'$ to be the entire set of jobs. When we remove a job $j$ from $\cA'$, the job will be automatically removed from $\cA'_t$ for all $t$. Since $\cA$ will change over the preprocessing steps, we will let $\cA^i$ denote the current set $\cA'$ just after completing $i^{th}$ step; $\cA^i_t$ is defined similarly. Also we will refer to  the quantity $\sum_{t: j \in \cA'_t} 1$ as job $j$'s residual flow time, which will change over preprocessing steps. A job's weighted residual flow time is similarly defined. Between steps, we will formally state some important changes made. For a subset $S$ of jobs, let $W(S)$ denote the total weight of jobs in $S$. 

\paragraph{Step 1. Discard jobs with small weighted flow time.}  For all jobs $j$ such that $w_j F^\cA_j \leq \frac{1}{2n} \max_{j'} w_{j'} F^\cA_{j'}$, remove them from $\cA'$.

\begin{proposition}
	\label{prop:1-2}
	$\sum_t  W(\cA^1_t)  = \sum_{j \in \cA^1} w_j F_j \geq  \frac{1}{2}\sum_{j} w_j F_j$.
\end{proposition}

\paragraph{Step 2. Group jobs with similar weights: odd/even classes.}  Let $w_{\max}$ denote the maximum job weight. A job $j \in \cA'$ is in class $\cC_h$, $h \geq 1$ if the job has weight in the range of $(w_{\max}/ n^{8h}, w_{\max} / n^{8(h-1)}]$. Classes $\cC_1, \cC_3, ..., $ are said to be odd, and the other classes even. Note that $\{\cC_h\}_{h \geq 1}$ is a partition of all jobs in $\cA'$. 

\paragraph{Step 3. Keep only odd [even] classes.} Between odd and even classes, keep the classes that give a larger total weighted flow time. More precisely, if $\sum_t \sum_{h: odd} W(\cC_h \cap \cA'_t) \geq \sum_t \sum_{h: even} W(\cC_h \cap \cA'_t)$, we keep only odd classes, i.e. remove all jobs in even classes from $\cA'$. Since the other case can be handled analogously, we will assume throughout the analysis that we kept odd classes.

\begin{proposition}
	\label{prop:3-4}
	$\sum_t  W(\cA^3_t) \geq \frac{1}{2} \sum_t  W( \cA^2_t) =  \frac{1}{2}\sum_t W( \cA^1_t) \geq \frac{1}{4} \sum_j w_j F_j$.
\end{proposition}

\paragraph{Step 4. Black out times with extra large weights for odd classes.}  We say that time $t$ has extra large weights for odd classes if $\sum_{h: odd} W(\cA'_t \cap \cC_h)  \leq \frac{1}{8}  W(\cA_t)$.  For all such times $t$, we remove all jobs from $\cA'_t$. Let's call such time steps ``global black out'' times, which we denote by $\cT_B$. 

	We now repeatedly remove jobs with a very small portion left -- we say that a job $j$ is almost-removed if the job $j \in \cA'$ has global black out times for at least $1 - 1/ n^4$ fraction of times during its window, $[r_j, C^\cA_j]$. For each time $t \in \cA'_t$, if $w_j \geq \frac{1}{n^2} W(\cA'_t)$, we remove all jobs from $\cA'_t$, and add the time $t$ to $\cT_B$. Also we remove the almost-removed job $j$ from $\cA'$. We repeat this step until we have no almost-removed job in $\cA'$. 

\begin{proposition}
	\label{prop:4-5-1}
	For any time $t \notin \cT_B$, we have $\sum_{h: odd} W(\cA^4_t \cap \cC_h) \geq \frac{1}{16} W( \cA_t)$. 
\end{proposition}
\begin{proof}
	Before the second operation of repeatedly removing almost removed jobs, we have that for all $t \notin \cT_B$, $\sum_{h: odd} W(\cA^4_t \cap \cC_h) \geq \frac{1}{8} W( \cA_t)$. Hence if a time step $t$ has ``survived" to the end of Step 4, i.e. $t \notin \cT_B$, it means that the quantity $W(\cA'_t)$ decreased by a factor of at most $1 - 1 / n^2$ when a job is removed by the second operation. Assuming that $n$ is greater than a sufficiently large constant, the proposition follows. 
\end{proof}

\begin{proposition}
	\label{prop:4-5-2}
	$\sum_t W(\cA^4_t)  \geq \frac{1}{4} \sum_t W( \cA^3_t)   \geq \frac{1}{16} \sum_j w_j F_j$.
\end{proposition}
\begin{proof}
	Define $V_1^j$ [$V_2^j$] to be the decrease of job $j$'s residual weighted flow time due to operation 1 [2]. Define $R$ to be all jobs $j$ in $\cA^3$ such that $\frac{1}{100} V^1_j \leq  V_j^2$ and $V^1_j + V^2_j \geq \frac{1}{4} F_j$.  We claim that the decrease in the total residual weighted flow time due to jobs in $\cA^3 \setminus  R$ is at most  $\frac{1}{2}  \sum_t W(A^3_t)$.
Firstly, the decrease due to jobs $j$ such that $V^1_j + V^2_j \leq \frac{1}{4} F_j$ is at most  $\frac{1}{4} \sum_t W( \cA^3_t)$. Secondly, the decrease due to jobs $j$ such that $\frac{1}{100} V^1_j \geq  V_j^2$ is at most the decrease due to the first operation in Step 4 times 101/100, which is upper bounded by $\frac{101}{100} \cdot \frac{1}{16} \sum_j w_j F_j$. This follows from the observation that the total weighted flow time that jobs in $\cA^3_t$ accumulate at black out times is at most $\frac{1}{16} \sum_j w_j F_j \leq \frac{1}{4} \sum_t W(A^3_t)$ by Proposition~\ref{prop:3-4}. Hence the claim follows. 
	
	We now focus on upper bounding the decrease in the total residual weighted flow time due to jobs $\cA^3  \cap R$. Observe the important property of jobs $j$ in $\cA^3 \setminus \cap R$: job $j$'s  residual flow time decreases by $\Omega( \frac{1}{n} F_j)$ when a certain job $j'$ is removed in the second operation in Step 4 (recall that removing $j'$ could lead to more times becoming blacked out, hence decreasing other jobs residual flow times). Motivated by this, we create a collection of rooted trees $\cF$, which describes which job  is most responsible for removing each job in $R$. We let $j'$ become $j$'s parent; if there are more than one such $j'$ we break ties arbitrarily. Note that roots in $\cF$ are those jobs $j$ whose residual flow time is $\frac{1}{n^4} F_j$ just before the second operation starts.  Now let's consider two jobs, $j'$ and $j$, a child of $j'$. We claim that $w_j F_j \leq O(\frac{1}{n^2}) w_{j'} F_{j'}$. Just before we remove job $j'$, job $j$ had residual flow time at least $\Omega(\frac{1}{n} F_{j'})$, and $j'$ had residual flow time at most  $\frac{1}{n^4} F_{j'}$. Further it must be the case that $w_{j'} \leq n^2 w_j$. This implies that job $j$'s weighted flow time is only a fraction of that of job $j'$, hence the claim follows.

	We complete the proof by charing a child's weighted flow time to its parent's weighted flow time. Eventually we will charge the total weighted flow time of all jobs in $\cF$ to that of the root jobs. 	Then it is easy to observe that the total weighted flow time of non-root jobs in $\cF$ is at most $O( 1/ n^2)$ times the total weighted flow time of the root jobs. Since we already upper bounded $(1 - 1/ n^4)$ times the total weighted flow time of root jobs by $\frac{1}{16} w_j F_j$, hence the decrease in the total residual weighted flow time due to jobs ($\cA^3 \cap R$ and $\cA^3 \setminus R$) in Step 3 is at most $\frac{n^4}{n^4 -1} \frac{1}{16} \sum_j w_j F_j$ + $\frac{1}{2}  \sum_t W(\cA^3_t) \leq \frac{3}{4}  \sum_t W(\cA^3_t)$ by Proposition~\ref{prop:3-4}. 
\end{proof}

\begin{proposition}
	\label{prop:4-5-3}
	For all jobs $j \in \cA^4$, we have $\sum_{t: j \in \cA_t^4} 1 \geq \frac{1}{n^4} F_j$.
\end{proposition}

\paragraph{Step 5. Find jobs with similar flow times within each odd  class $\cC_h$  that have the largest total weighted flow time.} Consider each odd class $h$. Observe that all jobs in $\cC_h$ have flow times  which are all within factor $2 n^9$ -- this follows from the facts that all jobs in the same class have weights, all within factor $n^8$, and all jobs have weighted flow times, all within factor $2n$. Let $F^{\max}_h$ denote the maximum flow time of all jobs in $\cC_h$. We say that a job $j \in \cC_h$ is in $u^{th}$-level if $F^\cA_j \in  (F^{\max}_h / 2^{u}, F^{\max}_h / 2^{u-1}]$, and denote all $u^{th}$-level jobs in $\cC_h$ as $\cC_{hu}$. Let $u^*_h$ be the $u$ that maximizes the total ``residual" weighted flow time of jobs in $\cC_{hu}$, i.e. $\sum_{j \in \cC_{hu}} \sum_{t: j \in \cA_t'} w_j$. Observe that $1 \leq u^*_h \leq \lceil \log 2 n^9 \rceil \leq  2^4\log n$ for all $n$ greater than a sufficiently large constant.  We remove all jobs $j \in \cC_h \setminus \cC_{h, u^*_h}$ from $\cA'$. 

\begin{proposition}
	\label{prop:5-6}
	For all $h \geq 1$, $\sum_t W(\cA^5_t \cap \cC_h) \geq \frac{1}{2^4 \log n} \sum_t W(\cA^4_t \cap \cC_h)$.
\end{proposition}

\paragraph{Step 6. Fit jobs in $\cC_{h, u^*_h}$ into a disjoint intervals of similar sizes.}  Consider each $h$. Define a set of intervals $\cL_{hu} := \{  [0, F^{\max}_h/2^{u}), [F^{\max}_h/2^{u}, 2  F^{\max}_h/2^{u}), [2 F^{\max}_h/2^{u}, 3  F^{\max}_h/2^{u}), ... \}$. Observe that jobs in $\cC_{h, u^*_h}$ have flow times similar to the size of an interval in $\cL_{h,u^*_h}$. Associate each job $j$ in $\cC_{h,u^*_h}$ with the interval $L$ in $\cL_{h,u^*_h}$ that maximizes $\sum_{t: j \in \cA'_t, t \in L} 1$, breaking ties arbitrarily. Here we trim out job $j$'s window so that it completely fits into the interval $L$, i.e. for all times $t \notin L$, remove job $j$ from $\cA'_t$. We let $L_j$ denote the unique interval in $\cL_h$ into which we fit  job $j$. Here we will refer to $L_j$ as job $j$'s laminar window. 

	We say that a job $j \in \cA'$ is left if $L_j$'s left end point lies in job $j$'s window, $[r_j, C_j]$. All other jobs in $\cA'$ are said to be right. Between left and right jobs, we keep the jobs that give a larger total remaining weighted flow time, remove other jobs from $\cA'$.

\begin{proposition}
	\label{prop:6-7}
	For all jobs $j \in \cA^6$, we have $\sum_{t: j \in \cA_t^6} 1 \geq \frac{1}{4 n^4} F_j$. 
\end{proposition}
\begin{proof}
There are at most 4  intervals in $\cC_{h, u^*_h}$ that intersect job $j$'s window. This, together with Proposition~\ref{prop:4-5-3} implies the claim. 
	
\end{proof}

\begin{proposition}
	\label{prop:6-7-2}
	For all odd $h \geq 1$, $\sum_t W(\cA^6_t \cap \cC_h) \geq \frac{1}{2^3} \sum_t W(\cA^5_t \cap \cC_h) \geq \frac{1}{2^7  \log n} \sum_t W(\cA^4_t \cap \cC_h)$. 
\end{proposition}

\paragraph{Step 7. Define exclusive active times for each odd class.}  We say that a class $\cC_h$ is lower than $\cC_{h'}$ if $h <  h'$. Observe that jobs in a lower odd class have flow time significantly smaller than jobs in a higher odd class. Consider odd classes $\cC_h$ in increasing order of $h$. For each job $j$, and for all times $t$ where a smaller odd class job $j'$ in $\cA'$ is processed ($j' \in \cA_t$),  remove job $j$ from $\cA'_t$. This makes different odd classes have disjoint ``active" times -- for any two jobs $j, j'$ from different odd classes, and for all times $t$, it is the case that $j \notin \cA'_t$ or  $j' \notin \cA'_t$.  For an odd class $h$, define $\cT_h$ to be the set of times $t$ where no job in odd class lower than $h$ is processed, and there is a job in $\cA'_t \cap \cC_h$.   We say that the times in $\cT_h$ are active for class $h$. Let $h^*_t$ denote the class that is active at time $t$ -- by definition there is at most one such class, and we let $h^*_t := 0$ if no such class exists. 

\begin{proposition}
	\label{prop:7-8}
	For all odd $h \geq 1$, $\sum_t W(\cA^7_t \cap \cC_h) \geq \frac{1}{2} \sum_t W( \cA^6_t \cap \cC_h) \geq \frac{1}{2^8 \log n} \sum_t W(\cA^4_t \cap \cC_h)$.
\end{proposition}
\begin{proof}
	 Observe that every job in $\cA^1$ of a  lower class than $j$ has a  flow time at most $2n/ n^8$ times that of job $j$ -- two jobs in different odd groups have different weights which differ by a factor of at least $n^8$ while all jobs in $\cA^1$  have the same weight flow time within factor $2n$. Hence jobs in smaller classes can create at most $\frac{2n^2}{n^8} F_j$ additional inactive time slots for job $j$, which are negligible compared to $\sum_{t: j \in \cA_t^6} 1 \geq \frac{1}{4 n^4} F_j$ (See Proposition~\ref{prop:6-7}). The proposition follows assuming that $n$ is greater than a sufficiently large constant. 
\end{proof}

\paragraph{Step 8. Black out times with extra large weights for $\cC_{h^*_t}$ at each time $t$.}  Consider each time $t$. We say that time $t$ has extra large weights for class $h^*_t$ if $\sum_{j \in \cA' \cap \cC_h} w_j   \geq 2^{9}\log n \sum_{j \in \cA^4_t \cap \cC_h} w_j$. For all such times $t$, we remove all jobs $j \in \cC_{h^*_t}$ from $\cA'_t$. Let's call such time steps ``black out'' times for class $\cC_{h^*_t}$.  Also remove those times from $\cT_{h^*_t}$.

\begin{proposition}
	\label{prop:8-1}
	For all odd $h \geq 1$, 

			$\sum_t W( \cA^8_t \cap \cC_{h^*_t}) \geq \frac{1}{2} \sum_t W( \cA^7_t \cap \cC_{h^*_t}) \geq \frac{1}{2^9 \log n} \sum_t W(\cA^4_t \cap \cC_{h^*_t}) \geq \frac{1}{2^{14} \log n} \sum_j w_j F_j.$	

\end{proposition}
\begin{proof}
	From Proposition~\ref{prop:7-8}, we have that $\sum_t W(\cA^7_t \cap \cC_{h^*_t})\geq \frac{1}{2^{8} \log n } \sum_t W(\cA^4_t \cap \cC_{h^*_t}) $. At every black out time, the total weight of jobs in $\cA^8_t \cap \cC_{h^*_t}$ is at most $\frac{1}{2^{9} \log n}$ times the total weight of jobs in $\cA^4_t \cap \cC_{h^*_t}$. This implies that in Step 8,  the decrease in the quantity $\sum_t W(\cA'_t \cap \cC_{h^*_t})$ is at most $\frac{1}{2^{9} \log n}$ times $ \sum_t W(\cA^4_t \cap \cC_{h^*_t})$, hence the second inequality follows. 
	We now show the last inequality. We can prove that $\sum_t W(\cA^4_t \cap \cC_{h^*_t}) \geq \frac{1}{2} \sum_t \sum_{h: odd} W(\cA^4_t \cap \cC_{h})$ (See the proof of Proposition~\ref{prop:8-2}). Then no jobs in even class contribute to the right-hand-side quantity (See Step 3), we have $\sum_t \sum_{h: odd} W(\cA^4_t \cap \cC_{h}) = \sum_t \sum_{h} W(\cA^4_t \cap \cC_{h})$. This, together with Proposition~\ref{prop:4-5-2}, completes the proof. 	
\end{proof}

\begin{proposition}
	\label{prop:8-2}
	At all times $t \notin \cT_B$, $W(\cA^8_t \cap \cC_{h^*_t})  \geq \frac{1}{ 2^{9}\log n} W(\cA^4_t \cap \cC_{h^*_t}) \geq  \frac{1}{ 2^{14}\log n} W(\cA_t)$.
\end{proposition}
\begin{proof}
	We focus on the second inequality since the first inequality is obvious. 	By definition of $h^*_t$, we know that no job in odd class smaller than $h^*_t$ is processed by $\cA$ at time $t$. Also we know that there is at least one job in $\cA^8_t \cap \cC_{h^*_t}$ whose weight is $n^4$ times larger than any job in a higher odd class. Hence we have that $W(A^4_t \cap \cC_{h^*_t})   \geq \frac{1}{2} \sum_{h: odd} W(\cA^4_t \cap \cC_{h})$.  Then the second inequality immediately follows from Proposition~\ref{prop:4-5-1}. 
\end{proof}

This completes the description of all preprocessing time steps. At this point, let's recap what we have obtained from these  preprocessing steps. Proposition~\ref{prop:8-2} says that at each time $t \notin \cT_B$, we can just focus jobs in one class, and furthermore all those jobs are at the same level and the level is the same at all times for the class. This is because at each time $t$, the total weight of those jobs at the same level in a class  is at least $\Omega( 1 / \log n)$ times the total weight of all jobs alive at the time. Also since we have fit jobs at the same level into intervals of the same size, and kept say only left jobs, we will be able to pretend that those jobs arrive at the same time. Then the analysis basically reduces to that of the weighted completion time objective.

\medskip
We are almost ready to set dual variables. Let $q_{jt}$ denote the size of job $j$ processed at time $t$. For each $h$, we define $\zeta_{t}(h)$ as follows. Let $L$ be the unique interval in $\cL_{h, u^*_h}$ such that $t \in L$.  
Let $\zeta_{t}(h)$ denote the weighted median from the multiset  $M(h, L):= \{\frac{q_{jt}}{p_j} \; | \; j \in \cA_t^8 \cap \cC_h\}$ -- here the median is taken assuming that the quantity $\frac{q_{jt}}{p_j}$ has $w_j$ copies in the multiset $M(h,L)$. As before, we set dual variables using the optimal solution $\vec{x}_t$ of $\cppf$, and its dual variables $\vec{y}_t$. Recall that each time step $t$ is active for at most one class $h$ which is denoted as $h^*_t$; if no such class exists, let $h^*_t= 0$ and $\zeta_{t}(h^*_t) := -1$. Define, 
\begin{align*}
	\alpha_{jt} &:= 
		\begin{cases} 	 
			 w_j   	&\quad \forall t, j \in \cC_{h^*_t}  \cap \cA^8_t \mbox{ s.t. }  \frac{q_{jt}}{p_j} \leq \zeta_{t}(h^*_t) \\
			0     &\quad \mbox{otherwise}					 
		\end{cases} 	
\end{align*}
\noindent and let
\begin{align*}
	&& \alpha_j &:= \sum_{t} \alpha_{jt} && \forall j  
\end{align*}

We continue to define $\beta_{dt}$. We will first define $\beta_{dt}(h)$ for each odd $h$, and will let $\beta_{dt} := \sum_{h: odd} \beta_{dt}(h)$ for all $d,t$.
Consider any odd $h$ and $L \in \cL_{h, u^*_h}$. Then for all times $t \in L$, define

\begin{align*}
	\beta_{dt}(h) &:=  \frac{1}{s} \sum_{t' \geq t, t' \in  \cT_h \cap L \setminus \cT_B }  \zeta_{t'}(h) y^*_{dt'} 
\end{align*}
Also for all odd $h$ and $t$ such that there is no $L \in \cL_{h, u^*_h}$ with $t \in L$, we let $	\beta_{dt}(h): = 0$. This completes setting dual variables.

\medskip
As before, we will first lower bound the objective of $\dual_s$. We start with lower bounding the first part in the objective. 

\begin{lemma}
	\label{lem:alpha-sum-flow}
	$\sum_j  \alpha_j \geq \frac{1}{2} \sum_t W( \cA^8_t \cap \cC_{h^*_t})$.
\end{lemma}
\begin{proof}
	 For any time $t$, it is easy to see from the definition of $\alpha_{jt}$ that the quantity $\sum_{j}  \alpha_{jt}$ is at least $\frac{1}{2} � W( \cA^8_t \cap \cC_{h^*_t})$. 

\end{proof}

In the following lemma we  lower bound the second part $\sum_{d,t} \beta_{dt}$ in the $\dual_s$ objective. The proof is very similar to that of Lemma~\ref{lem:obj-beta}. 

\begin{lemma}
	\label{lem:beta-sum-h-flow}
	For all odd $h$ and time $t \in \cT_h \setminus \cT_B$, $ \sum_d \beta_{dt}(h) \leq \frac{2^{20} \log n  }{s}  W( \cA^8_t \cap \cC_{h^*_t})$. 
\end{lemma}
\begin{proof}
	Consider any fixed odd $h$ and $L \in \cL_{h, u^*_h}$. Recall that amongst jobs $j \in \cC_{h, u^*_h}$ such that $L_j = L$, we kept left or right jobs in Step 6, We assume that we kept left jobs since the other case can be handled similarly. Consider any $t \in \cT_h \setminus \cT_B$.  Let $t_R$ denote the right end point of $L$. It is important to note that jobs contributing to $\cA^8_{t'} \cap \cC_h$ are consistent: the set of such jobs can only decrease in in time $t'  \in \cT_h \cap L \setminus \cT_B$ -- so does $W(\cA^8_{t'} \cap \cC_h)$.   We partition the time interval $[t, t_R)$ into subintervals  $\{M_k\}_{k \geq 1}$ such that the the quantity  $W(\cA^8_{t'} \cap \cC_h)$ at all times $t'$ during in $M_k$ lies in the range $\Big((\frac{1}{2})^{k} W( \cA^8_t \cap \cC_{h}), (\frac{1}{2})^{k-1} W( \cA^8_t \cap \cC_{h}) \Big]$. Now consider any fixed $k \geq 1$. We upper bound the contribution of $M_k$ to $\sum_{d} \beta_{dl}(h)$, that is $\frac{1}{s} \sum_{t'  \in M_k \cap \cT_h \setminus \cT_B} \sum_d \zeta_{t'}(h)  y^*_{dt'}$. Towards this end, we first upper bound $\sum_{t' \in M_k \cap  \cT_h \setminus \cT_B} \zeta_{t'}(h) \leq 4$. The key idea is to focus on the total weighted throughput processed during $M_k$. Job $j$'s fractional weighted throughput at time $t'$ is defined as $w_j\frac{q_{jt'}}{p_j}$, which is job $j$'s weight times the fraction of job $j$ that is processed at time $t'$; recall that $q_{jt'}$ denotes the size of job $j$ processed at time $t'$.

\begin{align*}
	\sum_{t' \in M_k\cap  \cT_h \setminus \cT_B} \zeta_{t'}(h)
	&\leq \sum_{t' \in M_k \cap \cT_h \setminus \cT_B} \Big(2 \sum_{j \in  \cA^8_t \cap \cC_h} \frac{w_j}{   W( \cA^8_{t'} \cap \cC_{h}) }\Big)  \cdot \oneindi \Big( \frac{q_{jt'}}{p_{j}}
	 \geq \zeta_{t'}(h) \Big)  \cdot  \frac{q_{jt'}}{p_j} \\
	 &\leq 2 \frac{1}{(1/2)^k W( \cA^8_t \cap \cC_{h}) } \sum_{t' \in M_k  \cap \cT_h \setminus \cT_B}   \sum_{j \in  \cA^8_t \cap \cC_h}  w_{j}  \ \frac{q_{jt'}}{p_j} \\
	&\leq 2 \frac{1}{(1/2)^k W( \cA^8_t \cap \cC_{h}) }   (1/2)^{k-1} W( \cA^8_t \cap \cC_{h})  
	=4
\end{align*}
	The first inequality follows from the definition of $\zeta_{t'}(h)$: for jobs $j$ with total weight at least half the total weight of jobs in $\cA^8_{t'} \cap \cC_h$, $\frac{q_{jt'}}{p_j} \geq \zeta_{t'}(h)$. The second inequality is due to the fact that $W( \cA^8_{t'} \cap \cC_{h}) \geq (\frac{1}{2})^{k} W( \cA^8_t \cap \cC_{h})$ for all times $t' \in M_k$. The last inequality follows since the total weighted throughput that can be processed during $M_k$ is upper bounded by the quantity $W( \cA^8_{t'} \cap \cC_{h})$ at the earliest time $t' \in M_{k'} \cap \cT_h \setminus \cT_B$, which is at most $(\frac{1}{2})^{k-1} W( \cA^8_t \cap \cC_{h})$. 
	We are now ready to complete the proof. 
\begin{align*}	
	 \sum_d  \beta_{dt}(h) 
	&= \frac{1}{s} \sum_{t' \geq t, t' \in \cT_h \setminus \cT_B} \zeta_{t'}(h) y_{dt'}^*  = \frac{1}{s} \sum_{k \geq 1} \sum_{t' \in M_k  \cap \cT_h \setminus \cT_B} \zeta_{t'}(h) \sum_{d}  y^*_{dt'} \\
	&= \frac{1}{s} \sum_{k \geq 1} \sum_{t' \in M_k  \cap \cT_h \setminus \cT_B} \zeta_{t'}(h) W( \cA_{t'}) \qquad \qquad \mbox{ [By Lemma~\ref{lem:sumdual}]}\\
	&= \frac{2^{17} \log n}{s} \sum_{k \geq 1} \sum_{t' \in M_k  \cap \cT_h \setminus \cT_B} \zeta_{t'}(h) W( \cA^8_{t'} \cap \cC_h) \qquad \mbox{ [By Proposition~\ref{prop:8-2}]}\\
	&= \frac{2^{17} \log n}{s} \sum_{k \geq 1} 4 (1/2)^{k-1} W( \cA^8_{t} \cap \cC_h) \qquad \mbox{ [By definition of $M_{k}$ and the fact $\sum_{t' \in M_k} \zeta_{t'}(h) \leq 4$]}\\
	&\leq \frac{2^{20} \log n}{s}  W( \cA^8_{t} \cap \cC_h)
\end{align*}

From Lemma~\ref{lem:alpha-sum-flow} and Lemma~\ref{lem:beta-sum-h-flow}, we derive that the $\dual_s$ objective is at least $\Omega(\sum_t W( \cA^8_t \cap \cC_{h^*_t}))$ with $s = 2^{22} \log n$. By Proposition, we conclude that the $\dual_s$ objective is $\Omega( 1/ \log n)$ times the total weighted flow time. 
\end{proof}

\medskip
To complete the proof of the upper bound result in Theorem~\ref{thm:flow}, it only remains to show that all dual constraints are satisfied.   Observe that the dual constraints (\ref{eqn:dual-3-flow}) and (\ref{eqn:dual-4-flow}) are trivially satisfied, hence we focus on dual constraint (\ref{eqn:dual-1-flow}). Recall that $q_{jt}$ denotes the size of job $j$ that is processed by the algorithm at time $t$. Note that $\sum_{t} q_{jt}  = p_j$.  

\begin{lemma}
	The dual constraint (\ref{eqn:dual-1-flow}) is satisfied. 
\end{lemma}
\begin{proof}
	We only need to consider $j \in \cA^8$. Say $j \in \cC_h$ ($h$ is odd as before). Also we only need to consider time $t$ before the interval $L_j$ ends. Observe that $j \in \cC_{h^*_t} \cap \cA^8$ only if $h = h^*_t$ and $t \notin \cT_B$.
\begin{align*}
	\frac{\alpha_j}{p_j} -  w_j \frac{t - a_j}{p_j} 
	&\leq \sum_{t' \geq t, t' \in \cT_h \setminus \cT_B} \frac{\alpha_{jt'}}{p_j}   && \mbox{[Since for all $t'$, $\alpha_{jt'} \leq w_j$]}\\	
	&\leq \sum_{t' \geq t, t' \in \cT_h \cap L_j \setminus \cT_B} \frac{\alpha_{jt'}}{p_j}   && \mbox{[Since $\alpha_{jt'} \leq w_j$ at all times $t' \notin L_j$]}\\	
	&= \sum_{t' \geq t, t' \in \cT_h \cap L_j \setminus \cT_B} w_j \frac{1}{p_j}  \cdot \oneindi \Big( \frac{q_{jt'}}{p_j} \leq \zeta_{t'}(h)\Big)&& \\
	&= \sum_{t' \geq t, t' \in \cT_h  \cap L_j \setminus \cT_B} w_j \frac{q_{jt'}}{p_j} \frac{1}{q_{jt'}} \cdot \oneindi \Big( \frac{q_{jt'}}{p_j} \leq  \zeta_{t'}(h)\Big) &&\\
	&= \sum_{t' \geq t, t' \in \cT_h  \cap L_j \setminus \cT_B} \frac{w_j}{x^*_{jt'}} \frac{q_{jt'}}{p_j}  \cdot \oneindi \Big( \frac{q_{jt'}}{p_j} \leq \zeta_{t'}(h)\Big)     && \mbox{[Since $q_{jt'} = x^*_{jt'}$]}\\	
	&\leq \sum_{t' \geq t, t' \in \cT_h  \cap L_j \setminus \cT_B} B_{\cdot j} \cdot \yv^*_{t'} \cdot \zeta_{t'}(h)    && \mbox{[By the KKT condition (\ref{eqn:kkt-2})]}\\	
	&= B_{\cdot j} \cdot \betav_t(h)    && \mbox{[By definition of $\betav_t(h)$]} \\
	&\leq B_{\cdot j} \cdot \betav_t    && \mbox{[By definition of $\betav_t$]}
\end{align*}
\end{proof}

%% file: 4.flowtime-lb.tex
\subsection{Lower Bound}
	\label{sec:flow-lower}
	
In this section, we prove the lower bound claimed in Theorem~\ref{thm:flow}. Towards this end, we will first prove a lower bound for makespan.
	
\begin{theorem}
	\label{thm:lb-makespan}
	Any deterministic non-clairvoyant algorithm is $\Omega( \sqrt {\log n})$-competitive for minimizing the makespan (the maximum completion time). Further, this is the case even when all jobs arrive at time $0$. 
\end{theorem}	

We prove that Theorem~\ref{thm:lb-makespan} implies the desired result. 

\begin{proofof}[the lower bound in Theorem~\ref{thm:flow}]
To see this, let $\cI_0$ denote the lower bound instance consisting of $N$ unweighted jobs that establishes the lower bound stated in Theorem~\ref{thm:lb-makespan}. By scaling, we can without loss of generality assume that the optimal (offline) makespan for this instance is 1. For any fixed $\eps >0$, we create $N^{1 / \eps}$ copies of instance $\cI_0$, $\{ \cI_e \}_{e \in \{0, 1, 2, ..., N^{1 / \eps} \}}$ where all jobs in  $\cI_e$ arrive at time $e$. There is a global constraint across all instances $\cI_e$ -- two jobs from different instances cannot be scheduled simultaneously. Then any deterministic non-clairvoyant algorithm that is given speed less than half the lower bound stated in Theorem~\ref{thm:lb-makespan} cannot complete all jobs in each $\cI_e$ within 2 time steps.  It is easy to see that there are at least $e/2$ jobs not completed during $[e, e+1)$ for any $e \in \{0, 1, 2, ..., N^{1 / \eps}\}$. Hence any deterministic online algorithm has total flow time $\Omega(N^{2 / \eps})$. In contrast, the optimal solution can finish all jobs within 1 time step, thus having total flow time $O(N \cdot N^{1 / \eps})$. This implies that the competitive ratio is $\Omega(n^{ (1  - \eps) / ( 1+ \eps)})$ where $n$ is the number of jobs in the entire instance concatenating all $\cI_e$, completing the proof of lower bound stated in Theorem~\ref{thm:flow}.
\end{proofof}

\smallskip
Henceforth, we will focus on proving Theorem~\ref{thm:lb-makespan}. Our lower bound instance comes from single source routing in a tree network with ``multiplicative speed propagation". As the name suggests, this network is hypothetical: a packet is transferred from node $v_a$ to $v_b$ at a rate equal to the multiplication of speeds of all routers that the packet goes through. To give a high-level idea of the lower bound, we fist discuss one-level tree, and then describe the full lower bound instance. Throughout this section, we refer to an arbitrary non-clairvoyant algorithm as $\cA$. 

\paragraph{One-level instance: \mbox{$\cI(1)$.}}
The root $\rho$ has $\Delta_1 := 4$ routers where only one router has $2$-speed and the other routers have $1$-speed. There are $\Delta_1$ packets (or equivalently jobs) to be routed to the root $\rho$. Only one job has size $2^2 - 1= 3$, and the other jobs have size $2^1 - 1 = 1$. Each job must be completely sent to the root, and it can be done only using routers. At any time, each router can process only one job. This setting can be equivalently viewed as the related machine setting, but we stick with this routing view since we will build our lower bound instance by multilayering this one-level building block. Obviously, the optimal solution will send the big job via the 2-speed router, thus having makespan $3/2$. Also intuitively, the best strategy for $\cA$ is to send all jobs at the same rate by equally assigning the 2-speed router to all jobs. Then it is easy to see that the online algorithm can complete all 1-size jobs only at time $\Delta_1 / ( \Delta_1 +1)$, and complete the 2-size job at time $\Delta_1/ ( \Delta_1 +1)+ 1 = 9/5$. Observe that giving more 1-speed routers does not give any advantage to the online algorithm since the main challenge comes from  finding the big job and processing it using a faster router.

\paragraph{Multi-level instance: \mbox{$\cI(h)$, $h \in [D = \Theta (\sqrt{ \log n})]$.}}
We create a tree $T_h$  with root $\rho$ where all jobs are leaves and each job $j$ can communicate with its parent node $u(j)$ via one of $u(j)$'s router, and the parent $u(j)$ can communicate with its parent node $u^{(2)}(j)$ via one of $u^2(j)$'s router, and so on; node/job $v$'s parent is denoted as $u(v)$. The tree $T_h$ has depth $h$. Every non-leaf node $v$ has $\Delta_h = 4^h$ children, which are denoted as $\cC_v$. Also each non-leaf node $v$ has a set $\cR_v$ of routers, whose number is exactly the same as that of $v$'s children, i.e. $|\cR_v| = |\cC_v| = \Delta_h$. All routers in $\cR_v$ have 1-speed except only one which has 2-speed. 

At any time, a feasible scheduling decision is a matching between routers $\cR_v$ and nodes $\cC_v$ for all non-leaf nodes $v$; when some jobs complete, this naturally extends to an injective mapping from $\cC_v$ to $\cR_v$. To formally describe this, let $g$ denote each feasible scheduling decision. Note that each feasible schedule $g$ connects each job to the root by a unique sequence of routers. Let $z_g$ denote the indicator variable. Let $\eta_j(g)$ denote the number of 2-speed routers on the unique path from $j$ to the root $\rho$ for $g$. When the schedule follows $g$ each job $j$ is processed at a rate of $2^{\eta_j(g)}$. We can formally describe this setting by $\psp$ as follows:
$$\poly = \Big\{ x_j \leq \sum_{g} 2^{\eta_j(g)} z_g \; \forall j; \quad \sum_g z_g \leq 1; \quad \vec{x} \geq 0; \quad \vec{z} \geq 0 \Big\}$$

We now describe job sizes, which are hidden to $\cA$.  Each non-leaf node $v$  has one special ``big" child amongst its $\Delta_h = 4^h$ children $\cC_v$ -- roughly speaking, a big child can have bigger jobs in its subtree. Note that $\cA$ is not aware of which child is big. For each node $v$ of depth $h-1$, define $\eta_v$ to be the  number of ``big" children on the path from $v$ to the root possibly including $v$ itself. Then $v$'s children/jobs, $\cC_v$ have the following sizes: for any integer $0 \leq k < \eta_v$, the number of jobs of size $2^{k+1} - 1$ is exactly $4^{h - \eta_v} ( 4^{\eta_v - k} - 4^{\eta_v -k  -1})$; for $k = \eta_v$, there are $4^{h - \eta_v}$  jobs of size $2^{\eta_v+1} -1$. Note that there is only one job of size $2^{h+1} - 1$ in $T_h$ and it is the biggest job in the instance.

The final instance will be $\cI(D)$. Since $\cI(D)$ has $4^{D^2}$ jobs, we have $D = \Theta(\sqrt {\log n})$.  For a visualization of the instance, see Figure~\ref{fig:lb-makespan}.

\begin{figure}[!ht]
  \centering
      \includegraphics[width=0.85\textwidth]{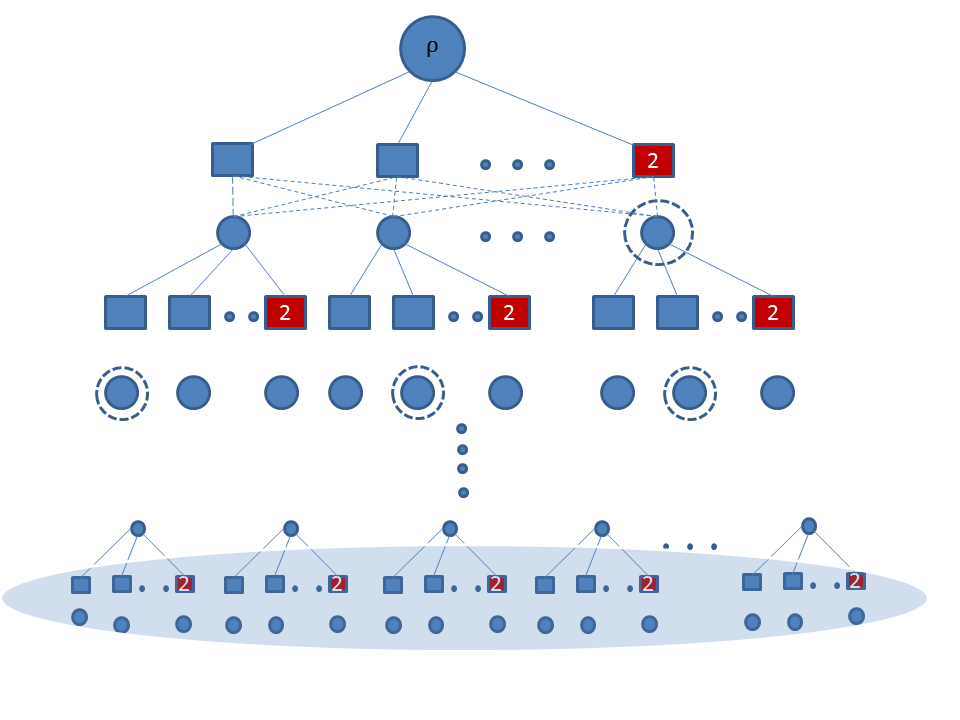} 	
  \caption{\label{fig:lb-makespan}Here, routers are represented by rectangles and jobs by circles. Height of the tree is $\Theta(\sqrt {\log n})$. Each job has $4^D$ children and for each level less than $D-1$ there is one big node, shown by dotted circle, which is hidden from the online algorithm. Each job has to be mapped to a router. The sizes of jobs at the last level depend on the number of big nodes on the path connecting a job to the root.}
\end{figure}

\begin{lemma}
	\label{lem:makespan-opt}
	There is an offline schedule that completes all jobs by time 2. 
\end{lemma}
\begin{proof}
	This is achieved by assigning each big node/job to the 2-speed router at all levels. This is possible since each non-leaf node has exactly one faster router and one big node/job. 	
	Consider any non-leaf node $v$ of depth $D-1$.  Since all jobs in $\cC_v$ have size at most $2^{\eta_v +1} - 1$, and all those jobs are processed at a rate of at least $2^{\eta_v}$, the claim follows. 
\end{proof}

We now discuss how $\cA$ performs for the instance $\cI(D)$. We first give a high-level overview of the adversary's strategy which forces $\cA$ to have a large makespan. Then, we will formalize several notions to make the argument clear -- the reader familiar with online adversary may skip this part. 

\paragraph{A high-level overview of the adversary's strategy.} As mentioned before, the main difficulty for the non-clairvoyant algorithm $\cA$ comes from the fact that $\cA$ does not know which jobs/nodes are big, hence cannot process big nodes using faster routers. This mistake will accrue over layers and will yield a gap $\Omega(D)$. To simplify our argument, we allow the adversary to {\em decrease} job sizes. That is,  at any point in time, the adversary observes the non-clairvoyant algorithm $\cA$'s schedule, and can remove any alive job. This is without loss of generality since the algorithm $\cA$ is non-clairvoyant, and can only be better off for smaller jobs. Obviously, this does not increase the optimal solution's makespan. 

Consider any node $v \neq \rho$.  Let us say that the subtree $T_{v'}$ rooted at $v'$ is big [small] if the node $v'$ is big [small]. 
If the node $v$ have used the unique 2-speed router in $\cR_{u(v)}$ for $1 / 2^{D+1}$ time steps, the adversary removes  the subtree $T_{v'}$ rooted at $v'$ (including all jobs in $T_{v'}$). 
We now show that at time $1/2$, the adversary still has an instance as effective as $\cI(D-1)$ -- by repeating this, the online algorithm will be forced to have a makespan of at least $D/2$. We claim two properties. 
\begin{enumerate}
	\item At time $1/2$, each alive non-leaf node has at least $4^{D-1}$ children. 
	\item Any job has been processed by strictly less than 1.
\end{enumerate}

The first property easily follows since each non-leaf node has $4^{D}$ children, and at most $2^D$ children are removed by time 1/2.  To see why the second property holds, consider any job $j$. Observe that each of $j$'s ancestor (including $j$ itself) used the 2-speed router only for $1 / 2^{D+1}$ time steps. Here the maximum processing for job $j$ can be achieved when $j$'s all ancestors use the 2-speed router simultaneously for $1 / 2^{D+1}$ time steps, which is most $1/2$. Also note that the length of time  job $j$ is processed by a combination of 1-speed routers only is strictly less than 1/2 time step. Hence the second property holds. 

Due to the second property, the online algorithm cannot find the big subtree incident to the root. This is because all subtrees incident to the root are indistinguishable to the algorithm by time 1/2, hence the adversary can pick any alive  one $T_v$ of those alive, and declare it is big. Likewise, for each alive non-leaf node $v'$, the adversary can keep alive the big child of $v'$. Hence the adversary can remove all nodes and jobs keeping only $4^{D-1}$ children including the big child for non-leaf node, and keeping only non-unit sized jobs. Also the adversary can pretend that all the remaining jobs have been processed exactly by one unit by decreasing job sizes. Observe that each alive job has remaining size $2^1 + 2^2 + ... + 2^k$ for some $k \geq 1$ Since $T_v$ is the only subtree incident to the root, we can assume that $\cA$ let $v$ use the 2-speed router from now on. This has the effect of decreasing each job's remaining size by half, and this exactly coincides with the instance $\cI(D-1)$. This will allow the adversary recurse on the instance $\cI(D-1)$, thereby making $\cA$'s makespan no smaller than $D/2 = \Omega(\sqrt{\log n})$.  This, together with, Lemma~\ref{lem:makespan-opt}, establish a lower bound $\Omega(\Delta) = \Omega (\sqrt{ \log n })$ for makespan, thus proving Theorem~\ref{thm:lb-makespan} and the lower bound claimed in Theorem~\ref{thm:flow}.

\bigskip
We formalize several notions (such as decreasing job sizes, indistinguishable instances) we used above to make the argument more clear. To this end it will be useful to define the collection $\cS: = \cS(0)$ of possible instances that the adversary can use.   The adversary will gradually decrease the instance space $\cS(t)$ depending on the algorithm's choice; $\cS(t)$ can only decrease in time $t$. Equivalently, the deterministic algorithm $\cA$ cannot distinguish between instances in $\cS(t)$ at the moment of time $t$, and hence must behave exactly the same by time $t$ for all the instances in $\cS(t)$. In this sense, all instances in $\cS(t)$ are indistinguishable to $\cA$ by time $t$.  All instances in $\cS$ follow the same polytope constraints for $\cI(D)$. There are two factors that make instances in $\cS$ rich.  The first factor is ``hidden" job IDs: Each non-leaf node $v$ has only one big child, and it can be any of its children $\cC_v$. In other words, this is completely determined by a function $\psi$ that maps each non-leaf node $v$ to one of $v$'s children, $\cC_v$. Consider any  fixed $\psi$. Then for each non-leaf node $v$ of depth $D-1$, the sizes that $v$'s children can have are fixed -- however, the actual mapping between jobs and job sizes can be arbitrary. So far, all instances can be viewed equivalent in the sense that they can be obtained from the same instance by an appropriate mapping. By a job $j$' ID, we mean the job in the common instance which corresponds to job $j$ in the common instance. The second factor  is ``flexible" job sizes.  Note that in $\cI(D)$,  a job $j$'s ID determines its size completely. This is not the case in $\cS$, and we let each job have any size up to the size determined by its ID. 

The adversary will start with set $\cS(0, D)$ -- here we added $D$ since the set is constructed from $\cI(D)$. The adversary's goal is to have $\cS(1/2, D)$ which essentially includes $\cS(0, D-1)$. By recursively applying this strategy, the adversary will be able to force $\cA$ have a makespan of at least $D/ 2$. As observed in Lemma~\ref{lem:makespan-opt}, for any instance in $\cS(0)$, all jobs in the instance can be completed by time 2 by the optimal solution, and this will complete the proof of Theorem~\ref{thm:lb-makespan}.

We make use of the two crucial properties we observed above. In particular, the second property ensures that at time $1/2$, any job ID mapping remains plausible in the solution set $\cS(1/2, D)$. 
Hence the adversary can choose any alive subtree $T_v$ incident to the root, and delete other sibling subtrees. From time $1/2$, any instance in $\cS(1/2, D)$ must satisfy the constraint that $v$ is big. 
The adversary now deletes all nodes/jobs in $T_v$ so that each node has $4^{D-1}$ children. Here the adversary can still choose any mapping (from each set of the alive siblings, the adversary can set any node/job to be big), and this 
has the same structure that $\cS(0, D-1)$ has regarding  the ``hidden Job ID" flexibility. Then, as mentioned before, by decreasing job sizes (more precisely, the corresponding instances are removed from $\cS(1/2, D)$) so that $\cS(1/2,D)$ becomes the same as $\cS(0, D-1)$ -- the only difference is that the root of the big subtree in $\cS(1/2, D)$ is processed via 2-speed router, however this difference is nullified by the fact that  the each job in any instance in $\cS(1/2, D)$ has exactly the double size that the corresponding job in the corresponding instance in $\cS(0, D-1)$ has. This allows the adversary to apply his strategy recursively. Hence we derive the following lemma which completes the proof of Theorem~\ref{thm:lb-makespan}  and the lower bound in Theorem~\ref{thm:flow}.

\begin{lemma}
	For any instance in $\cS(D)$, there is a way to complete all jobs in the instance within time 2. In contrast, for any deterministic non-clairvoyant algorithm $\cA$, there is an instance in $\cS(D)$ for which $\cA$ has a makespan of at least $D/2$.
\end{lemma}

%% file: psp.bbl
\begin{thebibliography}{10}

\bibitem{ec2spot}
http://aws.amazon.com/ec2/spot-instances/.

\bibitem{hadoop}
http://hadoop.apache.org.

\bibitem{vmpack}
http://www.vmware.com/files/pdf/vmware-distributed-resource-scheduler-drs-ds-en.pdf.

\bibitem{AcharyaFZ95}
S.~Acharya, M.~Franklin, and S.~Zdonik.
\newblock Dissemination-based data delivery using broadcast disks.
\newblock {\em IEEE Pers. Commun.}, 2(6):50--60, Dec 1995.

\bibitem{Ahmad2012}
Faraz Ahmad, Srimat~T. Chakradhar, Anand Raghunathan, and T.~N. Vijaykumar.
\newblock Tarazu: optimizing mapreduce on heterogeneous clusters.
\newblock In {\em ASPLOS}, pages 61--74. ACM, 2012.

\bibitem{AksoyF98}
Demet Aksoy and Michael~J. Franklin.
\newblock "rxw: A scheduling approach for large-scale on-demand data broadcast.
\newblock {\em IEEE/ACM Trans. Netw.}, 7(6):846--860, 1999.

\bibitem{AnandGK12}
S.~Anand, Naveen Garg, and Amit Kumar.
\newblock Resource augmentation for weighted flow-time explained by dual
  fitting.
\newblock In {\em SODA}, pages 1228--1241, 2012.

\bibitem{AzarBFP13}
Yossi Azar, Umang Bhaskar, Lisa Fleischer, and Debmalya Panigrahi.
\newblock Online mixed packing and covering.
\newblock In {\em SODA}, pages 85--100, 2013.

\bibitem{AzarG11}
Yossi Azar and Iftah Gamzu.
\newblock Ranking with submodular valuations.
\newblock In {\em SODA}, pages 1070--1079, 2011.

\bibitem{BansalC09}
Nikhil Bansal and Ho-Leung Chan.
\newblock Weighted flow time does not admit o(1)-competitive algorithms.
\newblock In {\em SODA}, pages 1238--1244, 2009.

\bibitem{BansalCKL14}
Nikhil Bansal, Moses Charikar, Ravishankar Krishnaswamy, and Shi Li.
\newblock Better scalable algorithms for broadcast scheduling.
\newblock In {\em SODA}, pages 55--71, 2014.

\bibitem{BansalCS08}
Nikhil Bansal, Don Coppersmith, and Maxim Sviridenko.
\newblock Improved approximation algorithms for broadcast scheduling.
\newblock {\em SIAM J. Comput.}, 38(3):1157--1174, 2008.

\bibitem{BansalKN10}
Nikhil Bansal, Ravishankar Krishnaswamy, and Viswanath Nagarajan.
\newblock Better scalable algorithms for broadcast scheduling.
\newblock In {\em ICALP (1)}, pages 324--335, 2010.

\bibitem{Bonald06}
T.~Bonald, L.~Massouli{\'e}, A.~Prouti\`{e}re, and J.~Virtamo.
\newblock A queueing analysis of max-min fairness, proportional fairness and
  balanced fairness.
\newblock {\em Queueing Syst. Theory Appl.}, 53(1-2):65--84, June 2006.

\bibitem{Boyd}
Stephen Boyd and Lieven Vandenberghe.
\newblock {\em Convex Optimization}.
\newblock Cambridge University Press, New York, NY, USA, 2004.

\bibitem{ChadhaGKM09}
Jivitej~S. Chadha, Naveen Garg, Amit Kumar, and V.~N. Muralidhara.
\newblock A competitive algorithm for minimizing weighted flow time on
  unrelated machines with speed augmentation.
\newblock In {\em STOC}, pages 679--684, 2009.

\bibitem{ChanEP09}
Ho-Leung Chan, Jeff Edmonds, and Kirk Pruhs.
\newblock Speed scaling of processes with arbitrary speedup curves on a
  multiprocessor.
\newblock In {\em SPAA}, pages 1--10, 2009.

\bibitem{ColeGG}
R.~Cole, V.~Gkatzelis, and G.~Goel.
\newblock Mechanism design for fair division: allocating divisible items
  without payments.
\newblock In {\em ACM EC}, pages 251--268, 2013.

\bibitem{EdmondsCBD03}
Jeff Edmonds, Donald~D. Chinn, Tim Brecht, and Xiaotie Deng.
\newblock Non-clairvoyant multiprocessor scheduling of jobs with changing
  execution characteristics.
\newblock {\em J. Scheduling}, 6(3):231--250, 2003.

\bibitem{EdmondsIM11}
Jeff Edmonds, Sungjin Im, and Benjamin Moseley.
\newblock Online scalable scheduling for the $\ell_k$-norms of flow time
  without conservation of work.
\newblock In {\em ACM-SIAM Symposium on Discrete Algorithms}, 2011.

\bibitem{EdmondsP12}
Jeff Edmonds and Kirk Pruhs.
\newblock Scalably scheduling processes with arbitrary speedup curves.
\newblock {\em ACM Transactions on Algorithms}, 8(3):28, 2012.

\bibitem{FoxIM13}
Kyle Fox, Sungjin Im, and Benjamin Moseley.
\newblock Energy efficient scheduling of parallelizable jobs.
\newblock In {\em SODA}, pages 948--957, 2013.

\bibitem{FoxK13}
Kyle Fox and Madhukar Korupolu.
\newblock Weighted flowtime on capacitated machines.
\newblock In {\em SODA}, pages 129--143, 2013.

\bibitem{GandhiKPS06}
Rajiv Gandhi, Samir Khuller, Srinivasan Parthasarathy, and Aravind Srinivasan.
\newblock Dependent rounding and its applications to approximation algorithms.
\newblock {\em J. ACM}, 53(3):324--360, 2006.

\bibitem{GargK07}
Naveen Garg and Amit Kumar.
\newblock Minimizing average flow-time : Upper and lower bounds.
\newblock In {\em FOCS}, pages 603--613, 2007.

\bibitem{drf}
A.~Ghodsi, M.~Zaharia, B.~Hindman, A.~Konwinski, I.~Stoica, and S.~Shenker.
\newblock Dominant resource fairness: Fair allocation of multiple resource
  types.
\newblock In {\em NSDI}, 2011.

\bibitem{GuptaIKMP12}
Anupam Gupta, Sungjin Im, Ravishankar Krishnaswamy, Benjamin Moseley, and Kirk
  Pruhs.
\newblock Scheduling heterogeneous processors isn't as easy as you think.
\newblock In {\em SODA}, pages 1242--1253, 2012.

\bibitem{GuptaKP12}
Anupam Gupta, Ravishankar Krishnaswamy, and Kirk Pruhs.
\newblock Online primal-dual for non-linear optimization with applications to
  speed scaling.
\newblock In {\em WAOA}, pages 173--186, 2012.

\bibitem{HallSSW97}
Leslie~A. Hall, Andreas~S. Schulz, David~B. Shmoys, and Joel Wein.
\newblock Scheduling to minimize average completion time: Off-line and on-line
  approximation algorithms.
\newblock {\em Math. of Oper. Res.}, 22(3):513--544, 1997.

\bibitem{ImKMP14}
Sungjin Im, Janardhan Kulkarni, Kamesh Munagala, and Kirk Pruhs.
\newblock {\sc SelfishMigrate}: A scalable algorithm for non-clairvoyantly
  scheduling heterogeneous processors, Manuscript, 2014.

\bibitem{ImM12}
Sungjin Im and Benjamin Moseley.
\newblock An online scalable algorithm for average flow time in broadcast
  scheduling.
\newblock {\em ACM Transactions on Algorithms}, 8(4):39, 2012.

\bibitem{ImMP11}
Sungjin Im, Benjamin Moseley, and Kirk Pruhs.
\newblock A tutorial on amortized local competitiveness in online scheduling.
\newblock {\em SIGACT News}, 42(2):83--97, 2011.

\bibitem{ImNZ12}
Sungjin Im, Viswanath Nagarajan, and Ruben van~der Zwaan.
\newblock Minimum latency submodular cover.
\newblock In {\em ICALP (1)}, pages 485--497, 2012.

\bibitem{kirk}
Bala Kalyanasundaram and Kirk Pruhs.
\newblock Speed is as powerful as clairvoyance.
\newblock {\em JACM}, 47(4):617--643, 2000.

\bibitem{KMT}
F.~P. Kelly, A.~K. Maulloo, and D.~K.~H. Tan.
\newblock Rate control for communication networks: Shadow prices, proportional
  fairness and stability.
\newblock {\em The Journal of the Operational Research Society}, 49(3):pp.
  237--252, 1998.

\bibitem{KellyMW09}
F.P. Kelly, L.~Massouli{\'e}, and N.S. Walton.
\newblock Resource pooling in congested networks: proportional fairness and
  product form.
\newblock {\em Queueing Systems}, 63(1-4):165--194, 2009.

\bibitem{lee2011heterogeneity}
Gunho Lee, Byung-Gon Chun, and Randy~H Katz.
\newblock Heterogeneity-aware resource allocation and scheduling in the cloud.
\newblock In {\em Proceedings of the 3rd USENIX Workshop on Hot Topics in Cloud
  Computing, HotCloud}, volume~11, 2011.

\bibitem{bargaining}
J.~Nash.
\newblock The bargaining problem.
\newblock {\em Econometrica}, 18(2):155--162, 1950.

\bibitem{popa2012faircloud}
Lucian Popa, Gautam Kumar, Mosharaf Chowdhury, Arvind Krishnamurthy, Sylvia
  Ratnasamy, and Ion Stoica.
\newblock Faircloud: sharing the network in cloud computing.
\newblock In {\em ACM SIGCOMM}, pages 187--198, 2012.

\bibitem{PruhsST}
Kirk Pruhs, Jiri Sgall, and Eric Torng.
\newblock {\em Handbook of Scheduling: Algorithms, Models, and Performance
  Analysis}, chapter Online Scheduling.
\newblock 2004.

\bibitem{QueyranneS02}
Maurice Queyranne and Maxim Sviridenko.
\newblock A (2+epsilon)-approximation algorithm for the generalized preemptive
  open shop problem with minsum objective.
\newblock {\em J. Algorithms}, 45(2):202--212, 2002.

\bibitem{RobertS08}
Julien Robert and Nicolas Schabanel.
\newblock Non-clairvoyant scheduling with precedence constraints.
\newblock In {\em Proceedings of the nineteenth annual ACM-SIAM symposium on
  Discrete algorithms}, SODA '08, pages 491--500, 2008.

\bibitem{SchulzS97}
Andreas~S. Schulz and Martin Skutella.
\newblock Random-based scheduling: New approximations and lp lower bounds.
\newblock In {\em RANDOM}, pages 119--133, 1997.

\bibitem{hdfs}
K.~Shvachko, H.~Kuang, S.~Radia, and R.~Chansler.
\newblock The hadoop distributed file system.
\newblock In {\em IEEE MSST}, pages 1--10, 2010.

\bibitem{WS}
David~P. Williamson and David~B. Shmoys.
\newblock {\em The Design of Approximation Algorithms}.
\newblock Cambridge University Press, 2011.

\bibitem{wolf2010flex}
Joel Wolf, Deepak Rajan, Kirsten Hildrum, Rohit Khandekar, Vibhore Kumar, Sujay
  Parekh, Kun-Lung Wu, and Andrey Balmin.
\newblock Flex: A slot allocation scheduling optimizer for mapreduce workloads.
\newblock In {\em Middleware}, pages 1--20. Springer, 2010.

\bibitem{Wong88}
J.~Wong.
\newblock Broadcast delivery.
\newblock {\em Proc. IEEE}, 76(12):1566--1577, 1988.

\bibitem{Zaharia08}
Matei Zaharia, Andy Konwinski, Anthony~D. Joseph, Randy Katz, and Ion Stoica.
\newblock Improving mapreduce performance in heterogeneous environments.
\newblock In {\em OSDI}, pages 29--42, Berkeley, CA, USA, 2008. USENIX
  Association.

\end{thebibliography}
